\documentclass[12pt, draftclsnofoot, onecolumn]{IEEEtran}
\usepackage{amsmath, amssymb, cite}
\usepackage{mathtools}
\usepackage[small]{caption}
\usepackage{amsthm,multirow,color,amsfonts}
\usepackage{tabulary}
\usepackage{subfig}
\usepackage{graphicx}
\usepackage{setspace}
\usepackage{enumerate}

\title{Optimizing Content Caching to Maximize the Density of Successful Receptions in Device-to-Device Networking}
\author{Derya~Malak, Mazin Al-Shalash, and~Jeffrey~G.~Andrews
\thanks{Parts of the manuscript were presented at the 2014 IEEE Globecom Workshops \cite{Malak2014} and at the 2015 IEEE ICC Workshops \cite{Malak2015}.}	
\thanks{D. Malak and J. G. Andrews are with the Wireless and Networking Communications Group (WNCG), The University of Texas at Austin, Austin, TX 78701 USA (email: deryamalak@utexas.edu; jandrews@ece.utexas.edu). M. Al-Shalash is with Huawei Technologies, Plano, TX 75075 USA (e-mail: mshalash@huawei.com).}
\thanks{This research has been supported by Huawei. \hfill Manuscript last revised: {\today}.}}

\newtheorem{theo}{Theorem}
\newtheorem{defi}{Definition}
\newtheorem{assu}{Assumption}

\newtheorem{cor}{Corollary}

\newtheorem{lem}{Lemma}

\newcommand{\SINR}{\sf{SINR}}

\newcommand{\SNR}{\sf{SNR}}  
\newcommand{\DD}{\sf{D2D}}  
\newcommand{\BS}{\sf{BS}}

\DeclareMathOperator*{\argmax}{arg\,max}
\newcommand{\PPP}{\sf{PPP}} 
\newcommand{\PGFL}{\sf{PGFL}} 
 
\newcommand{\dsr}{\rm{DSR}} 
\newcommand{\DSR}{\sf{DSR}} 
\newcommand{\DSRs}{\sf{DSR_S}} 
\newcommand{\DSRp}{\sf{DSR_P}} 
\newcommand{\DSRg}{\sf{DSR_G}} 
\newcommand{\Q}{\rm Q}
\newcommand{\T}{\rm T}
\newcommand{\tot}{\rm tot}
\newcommand{\Rayleigh}{\sf Rayleigh}
\newcommand{\RHS}{\sf RHS}
\newcommand{\pcov}{p_{\rm cov}}
\newcommand{\PCOV}{\mathcal{P}_{\rm cov}}

\begin{document} 
\maketitle	
\begin{abstract}
Device-to-device ($\DD$) communication is a promising approach to optimize the utilization of air interface resources in 5G networks, since it allows decentralized opportunistic short-range communication.  For $\DD$ to be useful, mobile nodes must possess content that other mobiles want.  Thus, intelligent caching techniques are essential for $\DD$. In this paper we use results from stochastic geometry to derive the probability of successful content delivery in the presence of interference and noise.  We employ a general transmission strategy where multiple files are cached at the users and different files can be transmitted simultaneously throughout the network.  We then formulate an optimization problem, and find the caching distribution that maximizes the density of successful receptions ($\dsr$) under a simple transmission strategy where a single file is transmitted at a time throughout the network. We model file requests by a Zipf distribution with exponent $\gamma_r$, which results in an optimal caching distribution that is also a Zipf distribution with exponent $\gamma_c$, which is related to $\gamma_r$ through a simple expression involving the path loss exponent. We solve the optimal content placement problem for more general demand profiles under Rayleigh, Ricean and Nakagami small-scale fading distributions. Our results suggest that it is required to flatten the request distribution to optimize the caching performance. We also develop strategies to optimize content caching for the more general case with multiple files, and bound the $\dsr$ for that scenario.
\end{abstract}

	
\IEEEpeerreviewmaketitle

\section{Introduction}
\label{intro}
Wireless networks are experiencing a well-known ever-rising demand for enhanced high rate data services, in particular wireless video, which is forecast to consume around 70\% of wireless bandwidth by 2019 \cite{Cisco2015}.  Non-real-time video in particular is expected to comprise half of this amount \cite{Sarkissian2012}, and comprises large files that can be cached in the network.  Meanwhile, preliminary $\DD$ techniques have been standardized by 3GPP to allow decentralized file sharing and public safety applications \cite{LinMag2014}. $\DD$ is intriguing since it allows increased spatial reuse and possibly very high rate communication without increased network infrastructure or new spectrum, but is only viable when the mobile users have content that other nearby users want.  Thus, it is clear that smart content caching is essential for $\DD$.

Caching popular content is a well known technique to reduce resource usage, and increase content access speed and availability \cite{Wang2014}. Infrastructure-based caching can reduce delay and when done at the network edge, also reduce the impact on the backhaul network, which in many cases is the bottleneck in wireless networks \cite{Shanmugam2013}.  However, this type of caching does not reduce the demand on spectral resources.  To gain spectral reuse and increase the area spectral efficiency, the content must be cached on wireless devices themselves, which allows short-range communication which is independent of the network infrastructure. $\DD$ communication can enable proximity-based applications involving discovering and communicating with nearby devices \cite{Lin2013}. Synchronized distributed network architectures for $\DD$ communications are designed, e.g., FlashLinQ \cite{Wu2013FlashlinQ} and ITLinQ \cite{Naderializadeh2014JSAC}, and caching is shown to provide increased spectral reuse in $\DD$-enabled networks \cite{Naderializadeh2014}. Although order optimal solutions for optimal content placement is known under certain channel conditions \cite{2Ji2013,Gitzenis2013,Jeon2015}, it is not known how to best cache content in a $\DD$ network. Intuitively, popular content should be seeded into the users' limited storage resources in a way that maximizes the probability that a given $\DD$ device can find a desired file within its radio range.  Exploring this problem quantitively is the goal of this paper.

\subsection{Related Work}
Different aspects of $\DD$ content distribution are studied. Scalability in ad hoc networks is considered \cite{Cruz2013}, where decentralized algorithms for message forwarding are proposed by considering a Zipf product form model for message preferences. Throughput scaling laws with caching have been widely studied \cite{MaddahAli2013Journal,Golrezaei2012,Ji2014}. Optimal collaboration distance, Zipf distribution for content reuse, best achievable scaling for the expected number of active $\DD$ interference-free collaboration pairs for different Zipf exponents is studied \cite{Golrezaei2014}. With a heuristic choice (Zipf) of caching distribution for Zipf distributed requests, the optimal collaboration distance \cite{Golrezaei2014TWC} and the Zipf exponent to maximize number of $\DD$ links are determined \cite{Golrezaei2012}. However, in general, the caching pmf is not necessarily same as the request pmf. This brings us to the one of the main objectives in this paper, which is to find the best caching pmf that achieves the best density of successful receptions ($\dsr$) in $\DD$ networks. 

Under the classical protocol model of ad hoc networks \cite{Gupta2000TIT}, for a grid network model, with fixed cache size $M$, as the number of users $n$ and the number of files $m$ become large with $nM \gg m$, the order optimal\footnote{The order optimality in \cite{2Ji2013,3Ji2013} is in the sense of a throughput-outage tradeoff due to simple model used.} caching distribution is studied and the per-node throughput is shown to behave as $\Theta(M/m)$ \cite{2Ji2013,3Ji2013}. The network diameter is shown to scale as $\sqrt{n}$ for a multi-hop scenario \cite{Gitzenis2013}. It is shown that local multi-hop yields per-node throughput scaling as $\Theta(\sqrt{M/m})$ \cite{Jeon2015}.

Spatial caching for a client requesting a large file that is stored at the caches with limited storage, is studied \cite{Altman2013}. Using Poisson point process ($\PPP$) to model the user locations, optimal geographic content placement and outage in wireless networks 
are studied \cite{Blaszczyszyn2014}. The probability that the typical user finds the content in one of its nearby base stations ($\BS$)s is optimized using the distribution of the number of $\BS$s simultaneously covering a user \cite{Keeler2013}. Performance of randomized caching in $\DD$ networks from a $\dsr$ maximization perspective has not been studied, which we study in this paper.

Although the work conducted in \cite{Golrezaei2012,Golrezaei2014} focused on the optimal caching distribution to maximize the average number of connections, the system model was overly simplistic. They assumed a cellular network where each $\BS$ serves the users in a square cell. The cell is divided into small clusters. $\DD$ communications are allowed within each cluster. To avoid intra-cluster interference, only one transmitter-receiver pair per cluster is allowed, and it does not introduce interference for other clusters. In this paper, we aim to overcome these serious limitations using a more realistic $\DD$ network model that captures the simultaneous transmissions where there is no restriction in the number of $\DD$ pairs.

\subsection{Contributions}
\label{contri}
This paper develops optimal content caching strategies that aim to maximize the average density of successful receptions so as to address the demands of $\DD$ receivers. The contributions are as follows.

{\bf Physical channel modeling using $\PPP$.} We introduce the network model in Sect. \ref{Model}, in which the locations of the $\DD$ users are modeled as a homogeneous $\PPP$. Different from the grid-based model in \cite{2Ji2013,3Ji2013}, we consider the actual physical channel model. $\PPP$ modeling makes our analysis tractable because unlike the cluster-based model in \cite{Golrezaei2014TWC}, where only a pair of users are allowed to communicate in a square region, we require no constraint on the link distance and allow a random number of simultaneous transmissions. All analysis is for a typical mobile node which is permissible in a homogeneous $\PPP$ by Slivnyak's theorem \cite{Stoyan1996}. The interference due to simultaneously active transmitters, noise and the small-scale Rayleigh fading are incorporated into the analysis. Any transmission is successful as long as the Signal-to-Interference-plus-Noise Ratio ($\SINR$) is above a threshold. 

{\bf Density of successful receptions ($\dsr$).} We propose a new file caching strategy exploiting stochastic geometry and the results of \cite{Andrews2011}, and we introduce the concept of the density of successful receptions ($\dsr$). Although in this paper, we do not investigate the throughput-outage tradeoff as in \cite{2Ji2013,3Ji2013}, the $\dsr$ is closely related to the outage probability, obtained through the scaling of the coverage, i.e., the complement of the outage probability, with the number of receivers per unit area.

{\bf Optimal caching distribution to maximize the $\dsr$ for the sequential multi-file model.} We study a randomized transmission model for $\DD$ users with storage size $1$ in Sect. \ref{Model}. We propose techniques for randomized content caching based on the possible ways of prioritizing different files. In Sect. \ref{SingleFile}, we start with a baseline model with single file to determine the optimal fractions of transmitters $\gamma_1$ and receivers $\gamma_2$ in the $\DD$ network model with $\PPP$ distributed user locations that maximizes the $\dsr$. In Sect. \ref{MultipleFiles}, we consider the more general sequential multi-file transmission scenario, where we investigate the maximum $\dsr$ in terms of the optimal fractions of $\gamma_1$ and $\gamma_2$ derived in Sect. \ref{SingleFile}, to determine the $\dsr$, and optimize the caching pmf based on the randomized model.

{\bf Small-scale fading $\dsr$ results.} We formulate an optimization problem in Sect. \ref{SUF} to find the best caching distribution that maximizes the $\dsr$ under a simple transmission strategy where single file is transmitted at a time throughout the network, assuming user demands are modeled by a Zipf distribution with exponent $\gamma_r$. This scheme yields a certain fraction of users to be active at a time based on the distribution of the requests. In Sect. \ref{Rayleighresults}, we optimize the $\dsr$ of users for the multi-file setup, where the small-scale fading is Rayleigh distributed. We consider several special cases corresponding to 1) small but non-zero noise, 2) arbitrary noise and 3) an approximation for arbitrary noise allowing the path loss exponent $\alpha=4$. For case 1), we show that the optimal caching strategy also has a Zipf distribution but with exponent $\gamma_c=\frac{\gamma_r}{\alpha/2+1}$ where $\alpha>2$. For case 2), we show that the same result holds based on an approximation of the $\SINR$ coverage justified numerically in Sect. \ref{Rayleighresults}. This relation implies that $\gamma_c$ is smaller than $\gamma_r$, i.e., the caching distribution should be more uniform compared to the request distribution, yet more popular files should be cached at a higher number of $\DD$ users. For case 3), we obtain a distribution similar to Benford's law (detailed in Sect. \ref{Rayleighresults}) that optimizes the caching pmf. We also extend our results to the ``general request distributions", and show that cases 1) and 2) are also valid for Ricean and Nakagami fading distributions in Sect. \ref{Rayleighresults}. 

In general, the optimal $\dsr$ and the optimal caching distribution might not be tractable. Therefore, assuming the request and caching probabilities are known a priori, we weight the caching pmf to provide iterative techniques to optimize the $\dsr$ under different settings. We propose caching strategies that consider maximizing the $\dsr$ of the least desired file and of all files as detailed in Sect. \ref{opt-multi-file}. 

{\bf Optimal caching distribution to maximize the $\dsr$ for the simultaneous multi-file model.} In Sect. \ref{mostgeneralcase}, we extend our study to the simultaneous transmissions of different files  and define popularity-based and global strategies. The popularity-based strategy is in favor of the transmission of popular files and discards unpopular files. On the other hand, the global strategy schedules all the files simultaneously, which leads to lower coverage than the sequential model does. Optimization of the $\dsr$ in these cases is very intricate compared to the case of sequential modeling. Therefore, we numerically compare the proposed caching models in Sect. \ref{mostgeneralcase}, and observe that the optimal solutions become skewed towards the most popular content in the network. Thus, we infer that under different models, the optimal caching distribution may not be a Zipf distribution as also found in \cite{2Ji2013,Gitzenis2013,Jeon2015}.

{\bf Insights.} Our results show that the optimal caching strategy exhibits less locality of the reference (abbreviated as locality) compared to the input stream of requests, i.e., the demand distribution\footnote{The performance of demand-driven caching depends on the locality exhibited by the stream of requests. The more skewed the popularity pmf, (i) the stronger the locality and the smaller the miss rate of the cache\cite{Vanichpun2005thesis}, and (ii) good cache replacement strategies are expected to produce an output stream of requests exhibiting less locality than the input stream of requests \cite{Makowski2005chapter}. In \cite{Vanichpun2005thesis}, authors showed that (i) and (ii) hold for caches operating under random on-demand replacement algorithms.}. We also analyze the special case of $\alpha=4$ using a tight approximation for standard Gaussian $\Q$-function. Using this approach we show that the optimal caching distribution can be approximated by Benford's law, which is a special bounded case of Zipf's law \cite{Pietronero2001}. In Sect. \ref{Performance}, we validate that both Zipf distribution and Benford's law have very similar distributional characteristics, further validating the generality of the results. For the multiple file case, we extend our results by finding lower and upper bounds for the $\dsr$ in Sect. \ref{bounds}. Simulations show that the bounds are very accurate approximations for particular $\gamma_r$ values.

\section{System Model}
\label{Model}
We consider a mobile network model in which $\DD$ users are spatially distributed as a homogeneous $\PPP$ $\Phi$ of density $\lambda$, where a randomly selected user can transmit or receive information. 
\begin{figure}[t!]
\centering
\includegraphics[width=0.7\columnwidth]{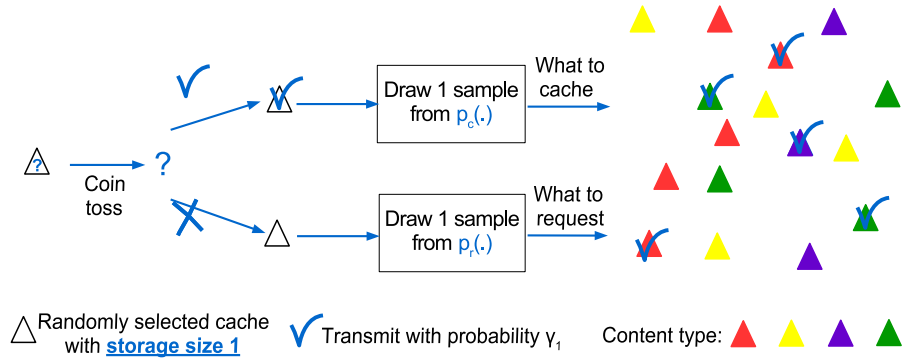}
\caption{\small{Randomized caching model.}}
\label{randomizedcaching}
\end{figure} 
In the multiple file scenario, the randomized caching model we propose is shown in Fig. \ref{randomizedcaching}. The model can be summarized as follows. At any time slot, only a fraction of  the $\DD$ users scheduled. Any user transmits with probability $\gamma_1$ and receives with probability $\gamma_2 = 1 - \gamma_1$ independently of other users. Each user has a cache with storage size $1$. If it is selected as a receiver at a time slot, it draws a sample from the request distribution $p_r(\cdot)$, which is assumed to be Zipf distributed. If it is selected as transmitter at a time slot, it draws a sample from the caching distribution $p_c(\cdot)$. The selection of request distribution and the optimization of caching distribution will be detailed in Sect. \ref{MultipleFiles}.  At any time slot, each receiver is scheduled based on closest transmitter association.

A system model for the $\DD$ content distribution network with multiple files is illustrated in Fig. \ref{modelmulti}. For multiple file case, different from the single file case, where the $\DD$ content distribution network is like a downlink cellular network since nearest transmitter has the content, a farther transmitter is often the one with the file required by the receiver.

\begin{figure}[t!]
\centering
\includegraphics[width=0.8\columnwidth]{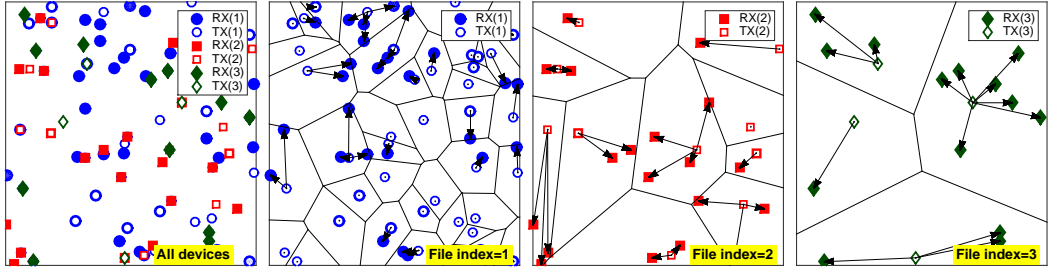}
\caption{\small{System model for $\DD$ users with multiple files. Each receiver is associated to its closest transmitter that contains the requested file, where $\rm TX(k)$ and $\rm RX(k)$ denote the set of transmitters and receivers corresponding to file $k$. For illustration purposes, different types are separated on the plot. However, transmissions of different files can occur simultaneously.}}
\label{modelmulti}
\end{figure}

General models for the multi-cell $\SINR$ using stochastic geometry were developed in \cite{Andrews2011}, where the downlink coverage probability was derived as:
\begin{eqnarray}
\label{pcov}
\pcov({\T},\lambda,\alpha)\triangleq\mathbb{P}[\SINR>\T]
=\pi\lambda \int_{0}^{\infty}{ e^{-\pi\lambda r\beta({\T},\alpha)-\mu \T\sigma^2 r^{\alpha/2}} \, \mathrm{d}r},
\end{eqnarray}
where $\beta({\T},\alpha)=\frac{2(\mu {\T})^{\frac{2}{\alpha}}}{\alpha}\mathbb{E}\big[g^{\frac{2}{\alpha}}(\Gamma(-2/\alpha,\mu {\T} g)-\Gamma(-2/\alpha))\big]$. The expectation is with respect to the interference power distribution $g$, the transmit power is $1/\mu$, and Signal-to-Noise Ratio ($\SNR$) is defined at a distance of $r=1$ and is ${\SNR} = 1/(\mu\sigma^2)$. A summary of the symbol definitions and important network parameters are given in Table \ref{tab-1}.
		
\begin{defi}{\bf Density of successful receptions ($\dsr$).} The performance of a randomly chosen receiver is determined by its $\SINR$ coverage. For the homogeneous $\PPP$ $\Phi$ with density $\lambda$, let $\gamma_1$ fraction of all users be the transmitter process $\Phi_t$, and $\gamma_2$ fraction of users be the receiver process $\Phi_r$, where $0<\gamma_1, \gamma_2<1$. The coverage probability of a randomly chosen receiver is $\pcov({\T},\lambda\gamma_1,\alpha)$, which is the same for all receivers, and the total average number of receivers is proportional to the density $\lambda\gamma_2$. Hence, the $\dsr$, which denotes the mean number of successful receptions per unit area, equals 
\begin{align}
\label{Ptotal}
{\DSR}=\lambda \gamma_2 \pcov({\T},\lambda\gamma_1,\alpha)=\lambda \gamma_2\Big(\pi\lambda\gamma_1 \int\nolimits_{0}^{\infty}{ e^{-\pi\lambda\gamma_1 r\beta({\T},\alpha)-\mu \T\sigma^2 r^{\alpha/2}} \, \mathrm{d}r}\Big),
\end{align}
where $\pcov({\T},\lambda\gamma_1,\alpha)$ is obtained by combining (\ref{pcov}) with the thinning property of the $\PPP$, i.e., $\Phi_t$, which is obtained through the thinning of $\Phi$, is a homogeneous $\PPP$ with density $\lambda\gamma_1$ \cite[Ch. 1]{BaccelliBook1}.
\end{defi}
We consider the generalized file caching problem in $\PPP$ networks where every user randomly requests or caches some files based on the availabilities. Our goal is to maximize the $\dsr$ in (\ref{Ptotal}) for single file and multiple files. We discuss the details of our optimization problem in Sects. \ref{SingleFile} and \ref{MultipleFiles}.

\section{$\dsr$ For a Single File}
\label{SingleFile}
We first assume that there is a single file in the network. The single file case is the baseline model for the more general multi-file model presented in Sect. \ref{MultipleFiles}. Sampled uniformly at random from the $\PPP$ $\Phi$, a fraction $\gamma_1$ of the users form the process $\Phi_t$ of the users possessing the file, and a fraction $\gamma_2$ of the users form the process $\Phi_r$ of the users who want the same file. The receivers communicate with the nearest transmitter while all other transmitters act as interferers, and each transmitter can serve multiple receivers. A receiver is in coverage when its $\SINR$ from its nearest transmitter is larger than some threshold $\T$. Given the total density of receivers is given by $\lambda\gamma_2$, and each receiver is successfully covered with probability $\pcov({\T},\lambda\gamma_1,\alpha)$, the $\dsr$, i.e., $\DSR$, is given by their product. In the single file scenario, since there is only 1 file being transmitted in the network, there is no caching pmf. Our objective in this section is to determine the optimal fractions of transmitters $\gamma_1$ and receivers $\gamma_2$ in the $\PPP$ network that maximizes the $\dsr$. In Sect. \ref{MultipleFiles}, we consider the multiple file transmission scenario, where we use the optimal fractions of transmitters and receivers $\gamma_1$ and $\gamma_2$, respectively, derived in this section, to determine the $\dsr$, and optimize the caching pmf based on the randomized model outlined in Sect. \ref{Model}. We formulate the following optimization problem to determine $\gamma_1$ and $\gamma_2$:
\begin{equation}
\label{optimization1}
\begin{aligned}
{\DSR^*}=&\underset{\gamma_1>0,\, \gamma_2>0}{\max}
& & \lambda \gamma_2 \pcov({\T},\lambda\gamma_1,\alpha) \\
& \hspace{0.6cm}\text{s.t.} 
& & \gamma_1+\gamma_2=a,\quad 0<a \leq 1,
\end{aligned}
\end{equation}
where $\pcov({\T},\lambda\gamma_1,\alpha)$ is the coverage probability of a typical user, and $a\leq 1$ is the total fraction of transmitting and receiving users in a $\PPP$ network $\Phi$ with density $\lambda$.

\begin{table}[t!]\footnotesize
\begin{center}
\setlength{\extrarowheight}{1.3pt}
\begin{tabular}{| l | c | }
\hline
{\bf Symbol} & {\bf Definition} \\ \hline
$\T$; $\alpha>2$ & $\SINR$ threshold; Path loss exponent  \\ \hline
$\gamma_1$; $\gamma_2$ & Fraction of transmitting users; fraction of receiving users \\ \hline
$\Phi$; $\Phi_t$; $\Phi_r$ & Homogeneous $\PPP$ of all $\DD$ users; transmitter process; receiver process \\ \hline
$\lambda$; $\lambda_t$ & Intensity of $\Phi$; intensity of $\Phi_t$  \\ \hline
$\mu^{-1}$; $\sigma^2$ & The constant transmit power; Noise variance \\ \hline
$g\sim\exp(\mu)$ & Interference power distribution\\ \hline
$\gamma_r$; $\gamma_c$ & Zipf request parameter; Zipf caching parameter \\ \hline
$M$; $1$ & Size of the file catalog; storage size of any user \\ \hline
$p_r(\cdot)$; $p_c(\cdot)$ & Popularity pmf; caching pmf \\ \hline
$\pcov({\T},\lambda,\alpha)$ & Coverage probability for the sequential transmission model \\ \hline
$\PCOV({\T},\lambda,\alpha)$ & Coverage probability for the general transmission model \\ \hline
$\beta({\T},\alpha)$ & A function of interference in the exponent of $\pcov$  \\ \hline
$F_B(\cdot)$ & The pmf of the Benford's distribution  \\ \hline
$\DSR$ & Density of successful receptions \\ \hline
$\DSRs; \DSRp; \DSRg$ & Sequential; popularity-based; global model $\DSR$\\ \hline
$\Q$-function & The tail probability of the standard normal distribution \\ \hline
$\Theta(\cdot)$; $o(\cdot)$ & Big O notation; Little-o notation \\ \hline
\end{tabular}
\end{center}
\caption{\small{Important network parameters.}}
\label{tab-1}
\end{table}

\begin{lem}\label{txfrac}
The fraction of transmitters should be less than that of receivers, i.e., the solution of (\ref{optimization1}) satisfies the following relation: $\gamma_1<a/2<\gamma_2<a\leq 1$. 
\end{lem}
\begin{proof}
See Appendix \ref{App:AppendixFirst}.
\end{proof}

\begin{lem}\label{arbitrarynoiselemma}
The maximum $\dsr$ for arbitrary noise and $\alpha=4$ is given by
\begin{eqnarray}
{\DSR^*}={\lambda(a-\gamma_1)}\Big/{\Big(\frac{1}{\gamma_1}\left[\frac{1}{\gamma_1}-\frac{1}{a-\gamma_1}\right]\frac{{2\mu T\sigma^2}}{(\pi \lambda)^2 \beta({\T},4)}+\beta({\T},4)\Big)}.\nonumber
\end{eqnarray}
\end{lem}

\begin{proof}
See Appendix \ref{App:Appendix0}.
\qedhere
\end{proof}

\begin{cor}{\bf Low $\SNR$ case, $\alpha=4$.} As $\sigma^2\to\infty$, the coverage can be approximated as $\pcov({\T},\lambda,\alpha)=\mathbb{P}[\SINR>\T]\approx\mathbb{P}[\SNR>\T]=\pi\lambda \int_{0}^{\infty}{ e^{-\pi\lambda r-\mu \T\sigma^2 r^{\alpha/2}} \, \mathrm{d}r}$. Hence, the maximum $\dsr$ is given as
\begin{eqnarray}
{\DSR^*}={\lambda(a-\gamma_1)}\Big/{\Big(\frac{1}{\gamma_1}\left[\frac{1}{\gamma_1}-\frac{1}{a-\gamma_1}\right]\frac{{2\mu T\sigma^2}}{(\pi \lambda)^2}+1\Big)},
\end{eqnarray}
where optimal $\gamma_1$ satisfies $\frac{a-3a\gamma_1+3\gamma_1^2}{\gamma_1^3(a-\gamma_1)}=\frac{(\pi \lambda)^2}{{4\mu T\sigma^2}}$.
\end{cor}

\begin{cor}{\bf No noise (degenerative) case.} For no noise, $\pcov({\T},\lambda,\alpha)=\beta({\T},\alpha)^{-1}$. Maximum $\dsr$ for single file for $0<a\leq 1$, Rayleigh fading, no noise, and $\alpha > 2$ is ${\DSR^*}=\underset{\gamma_1>0}{\max}\,\, \lambda(a-\gamma_1)\frac{1}{\beta({\T},\alpha)}=\frac{\lambda (a-\gamma_1^*)}{\beta({\T},\alpha)}$, obtained for the optimal value of $\gamma_1$, i.e., $\gamma_1^*=\varepsilon>0$ so that there is one transmitter\footnote{In the no noise case the single file result is trivial. In multiple file case, there will be interference due to the simultaneous transmissions of multiple files, which will be discussed in Sect. \ref{MultipleFiles}.}.
\end{cor}

Next, we consider the low noise approximation of the success probability that is more easily computable than the constant noise power expression and more accurate than the no noise approximation for $\sigma^2= 0$. Using the expansion $\exp(-x)=1-x+o(x)$ for $\sigma^2\neq 0$ as $x\to 0$, the term $\pcov({\T},\lambda,\alpha)$ for small but non-zero noise case can be calculated after an integration by parts of (\ref{pcov}) as follows
\begin{eqnarray}
\pcov({\T},\lambda,\alpha)=\frac{1}{\beta({\T},\alpha)}-\frac{\mu \T\sigma^2\left(\lambda\pi\right)^{-\frac{\alpha}{2}}}{\beta({\T},\alpha)^{\frac{\alpha}{2}+1}}\Gamma\left(1+\frac{\alpha}{2}\right)+o\left(\sigma^2\right).\nonumber
\end{eqnarray}

\begin{lem}\label{maximumcoveragerayleighlemma}
The maximum $\dsr$ for a single file for $a=1$, Rayleigh fading, small noise is equal to
\begin{eqnarray}
{\DSR^*}=\frac{\lambda\alpha}{\beta({\T},\alpha)}\left[\frac{1}{\alpha}-\frac{(\gamma_1^*-1)}{\alpha+\gamma_1^*(2-\alpha)}o(\sigma^2)\right].\nonumber
\end{eqnarray}
\end{lem}

\begin{proof}
See Appendix \ref{App:Appendix1}.
\qedhere
\end{proof}

For $\alpha=4$, there is a closed form expression for $\beta({\T},4)$ as follows: $\beta({\T},4)=1+\sqrt{\T}\arctan(\sqrt{\T})$, which we use for the derivation of Lemma \ref{maxcoversmallnoise}.

\begin{lem}\label{maxcoversmallnoise}
The maximum $\dsr$ for small but non-zero noise and $\alpha=4$ is 
\begin{eqnarray}
{\DSR^*}=\frac{2\lambda(a-\gamma_1)}{(1+\sqrt{\T}\arctan(\sqrt{\T}))}\left[1-\frac{\mu {\T}\sigma^2a}{\mu {\T}\sigma^2(2a-\gamma_1)+o(\sigma^2)}\right]+o(\sigma^2).
\end{eqnarray}
\end{lem}

\begin{proof}
See Appendix \ref{App:AppendixA}.
\qedhere
\end{proof}

\begin{figure*}
\begin{minipage}[t]{.328\textwidth}
\centering
\includegraphics[width=\textwidth]{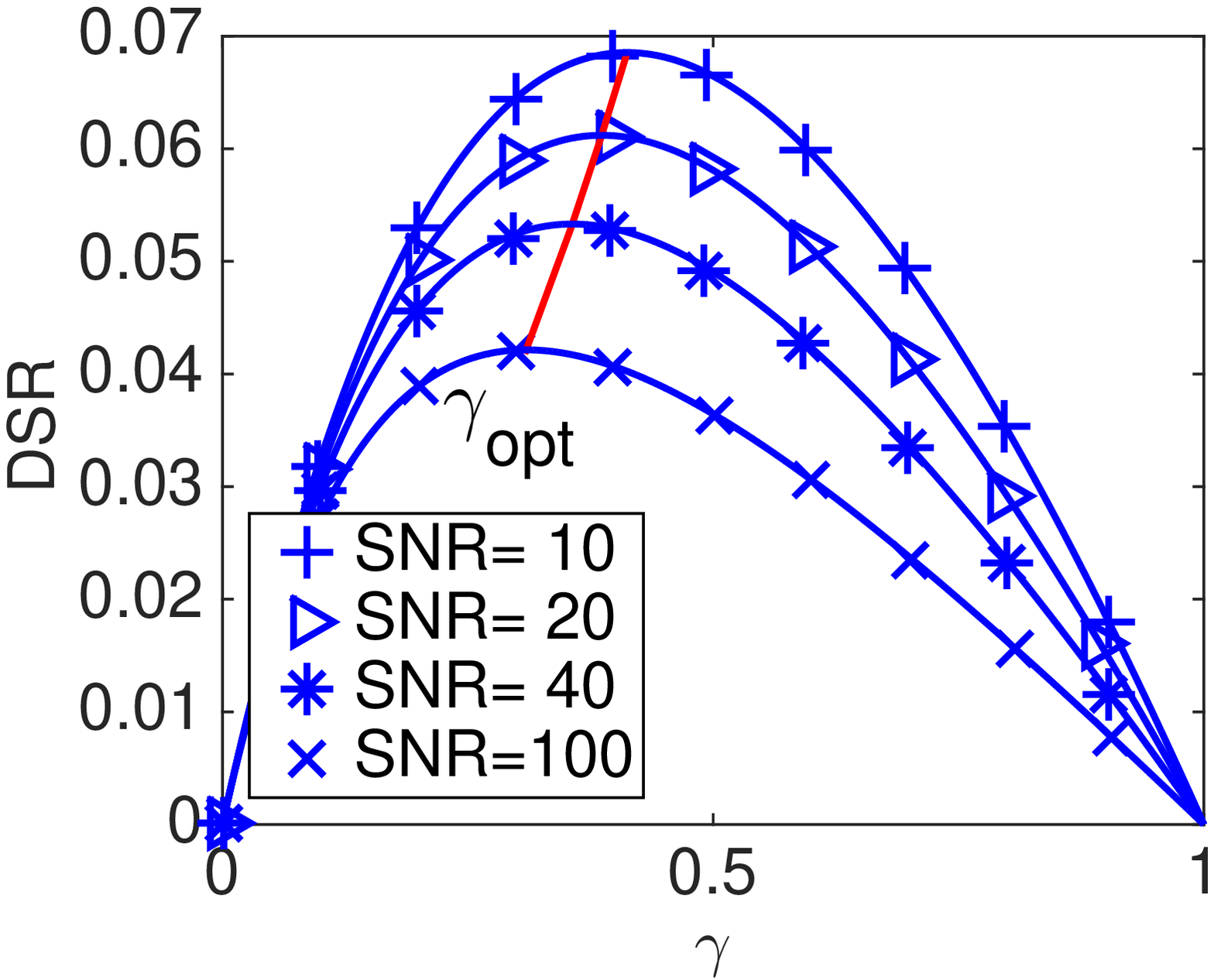}
\small{(a)}
\end{minipage}
\hfill
\begin{minipage}[t]{.328\textwidth}
\centering
\includegraphics[width=\textwidth]{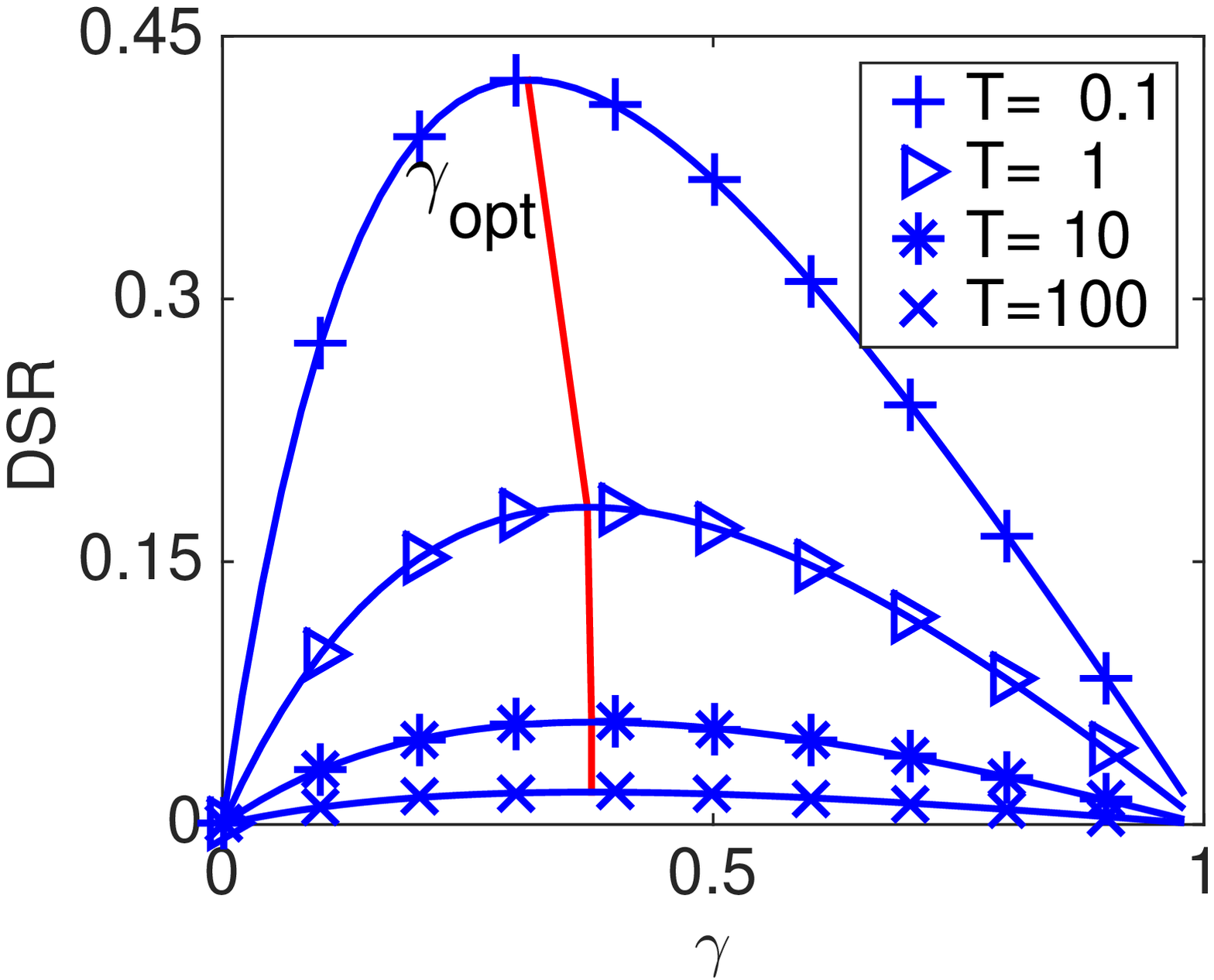}
\small{(b)}
\end{minipage}
\hfill
\centering{\begin{minipage}[t]{.328\textwidth}
\centering
\includegraphics[width=\textwidth]{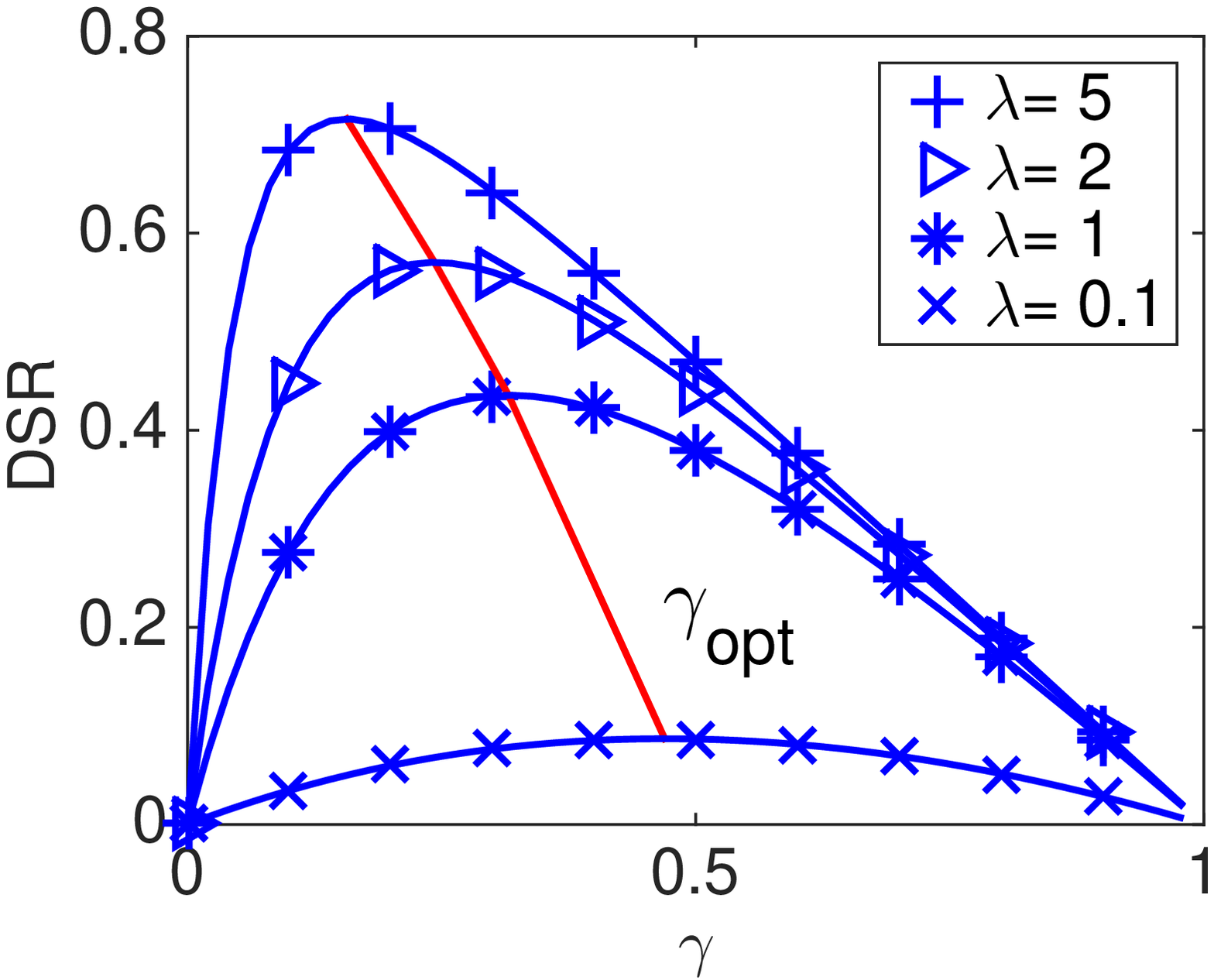}
\small{(c)}
\end{minipage}}
\caption{\small{$\dsr$ for single file versus $\gamma$ with respect to $\SNR$, $\T$ and $\lambda$. (a) $\dsr$, $\T=\SNR/2$, $\lambda$=0.1, where the dashed curves correspond to the respective Monte Carlo simulations, (b) $\dsr$, $\SNR=20$, $\lambda$=0.1, and (c)
$\dsr$, $\SNR=.1$, $\T=.05$.}\label{fig-all-1}}
\end{figure*}

{\bf Discussion.}
In Fig. \ref{fig-all-1} (a), we illustrate the relation between ${\DSR^*}$ and $\SNR$ for $\T=\SNR/2$, $\lambda$=0.1. To simplify the notation, we assume that $\gamma_1+\gamma_2=1$ and let $\gamma=\gamma_1$ and $\gamma_1^*=\gamma_{\mathrm{opt}}$. As $\SNR$ increases for $\T=\SNR/2$, the $\dsr$ decreases and $\gamma_{\mathrm{opt}}$ decreases. Note that the solid lines denote the simulation results for the $\PPP$ model. In Fig. \ref{fig-all-1} (b), the variation of ${\DSR^*}$ with respect to $\T$ for $\SNR=10$, $\lambda$=0.1 is shown. The coverage $\pcov({\T},\lambda\gamma_1,\alpha)$ is monotonically decreasing in $\T$ and a concave increasing function of $\gamma_1$. For increasing $\T$, the value of $\dsr$ becomes very small, and to maximize the $\dsr$, a higher fraction of the users should be transmitters (i.e., higher $\gamma_1$) to compensate the outage. For low $T$, to maximize the $\dsr$, the fraction of the receivers $\gamma_2$ should be higher. Therefore, as $\T$ decreases, the $\dsr$ increases and becomes right-skewed, but $\gamma_{\mathrm{opt}}$ decreases only slightly, which is negligible\footnote{This follows from the separability assumption of $\pcov({\T},\lambda\gamma_1,\alpha)$ in $\lambda\gamma_1$ and $T$, thus insensitivity of the $\dsr$ maximization problem to the value of $T$, which is further detailed in Assumption \ref{assumptionbeta} of Sect. \ref{Rayleighresults}, and verified in Appendix \ref{App:AppendixC}.}. Thus, we conclude that $\gamma_{\mathrm{opt}}$ is largely invariant to $\T$ and mainly determined by $\SNR$. In Fig. \ref{fig-all-1} (c), we show the variation of ${\DSR^*}$ with $\lambda$. The $\dsr$ increases with $\lambda$. On the other hand, $\gamma_{\mathrm{opt}}$ decreases as the density of users increases and transmissions from increased number of users cause high interference.

Although the single file case is trivial in the sense that it boils down to the optimization of the fractions of the transmitters and receivers that maximizes the $\dsr$, it is the baseline model for the multiple file case where the main objective is to determine the optimal caching distribution over the set of files. We discuss the multiple file setup next.

\section{Optimizing the $\dsr$ of the Sequential Serving Model with Multiple Files}
\label{MultipleFiles} 
We determine the optimal caching distribution for the transmitters to maximize the $\dsr$ for the sequential serving-based strategy, in which one type of file is transmitted at a time. Later, in Sect. \ref{mostgeneralcase}, we study the general case, where the transmissions of different files can take place simultaneously.

{\bf File Popularity Distribution.} To model the file popularity in a general $\PPP$ network, we use Zipf distribution for $p_r$, which is commonly used in the literature \cite{Golrezaei2014}. Then, the popularity of file $i$ is given by $p_r(i)={\frac{1}{i^{\gamma_r}}}\Big/{\sum\limits_{j=1}^{M}{\frac{1}{j^{\gamma_r}}}}$, for $i=1,\hdots, M$, where $\gamma_r$ is the Zipf exponent and there are $M$ files in total. The demand distribution $p_r\sim$Zipf$(\gamma_r)$ is the same for all receivers of the $\PPP$ model. 

\subsection{Sequential Serving-based Model}
\label{SUF}
In this model, only the set of transmitters having a specific file transmits simultaneously. Hence, this is the special case where only one file is transmitted at a time network-wide. This is illustrated in Fig. \ref{randomizedcaching} in Sect. \ref{Model}. If a user is selected as a receiver at a time slot, it draws a sample from the request distribution $p_r(\cdot)$, which is known. If any user is randomly selected as the transmitter at a time slot with probability $\gamma_1$, it draws a sample from the caching distribution $p_c(\cdot)$, which is not known yet. At any time slot, each receiver is scheduled based on closest transmitter association. According to this model, since file $i$ is available at each transmitter with $p_c(i)$, using the thinning property of the $\PPP$ \cite[Ch. 1]{BaccelliBook1}, the probability of coverage for file $i$ is 
\begin{align}
\label{thinning}
\pcov({\T},\lambda_t p_c(i),\alpha)
=\pi\lambda_t p_c(i)\int\nolimits_{0}^{\infty}{ e^{-\pi\lambda_t p_c(i)r\beta({\T},\alpha)-\mu \T\sigma^2 r^{\alpha/2}} \, \mathrm{d}r},
\end{align}
where $\lambda_t=\lambda\gamma_1$ is the total density of the transmitting users. 

Given that the requests are modeled by the Zipf distribution, our objective is to maximize the $\dsr$ of users for the sequential serving-based model, denoted by $\DSRs$ for a $\PPP$ model with density $\lambda$: 
\begin{equation}
\label{main-opt}
\begin{aligned}
& \underset{p_c}{\max}
& & \DSRs \\
& \text{s.t.}
& & \sum\limits_{i=1}^{M} p_c(i)=1; \quad p_r(i)={\frac{1}{i^{\gamma_r}}}\Big/{\sum\limits_{j=1}^{M}{\frac{1}{j^{\gamma_r}}}}, \quad i = 1, \ldots, M, 
\end{aligned}
\end{equation}
where ${\DSRs}=\lambda\gamma_2\sum\limits_{i=1}^M {p_r(i) \pcov({\T}, \lambda \gamma_1 p_c(i),\alpha)}$, the first constraint is the total probability law for the caching distribution, and the second constraint is the demand distribution modeled as Zipf with exponent $\gamma_r$, and $\gamma_2=1-\gamma_1$, and $M$ is the number of files. 

Note that $\pcov({\T}, \lambda \gamma_1 p_c(i),\alpha)$ in (\ref{main-opt}) is obtained for a sequential transmission or scheduling model and it is same as the formulation given in (\ref{pcov}) which follows from Theorem 1 of \cite{Andrews2011}. This model can be generalized to different scheduling schemes. For example, in Sect. \ref{mostgeneralcase}, we introduce a more general model where multiple files are simultaneously transmitted, and obtain a coverage expression $\PCOV({\T},\cdot,\alpha)$ that is different from $\pcov({\T}, \cdot,\alpha)$ in (\ref{main-opt}), which is detailed in Theorem \ref{mainProbCov} of Sect. \ref{mostgeneralcase}.

Similar to the optimal fractions of the transmitter and receiver processes calculated in Sect. \ref{SingleFile} for the single file case, optimal values of $\gamma_1$ and $\gamma_2=1-\gamma_1$ for multi-file case can be found by taking the derivative of (\ref{main-opt}) with respect to $\gamma_1$, which yields the following expression:
\begin{eqnarray}
\label{optimalgammaeq}
\sum\limits_{i=1}^M\lambda p_r(i)p_c(i)\Big\{\int_{0}^{\infty}{\Big[\frac{1}{\gamma_1}-\frac{1}{1-\gamma_1}-\pi\lambda p_c(i)\beta({\T},\alpha)r\Big] e^{-\pi\lambda\gamma_1 p_c(i) r\beta({\T},\alpha)-\mu {\T}\sigma^2 r^{\frac{\alpha}{2}}} \, \mathrm{d}r}\Big\}=0,
\end{eqnarray}
where optimal value of $\gamma_1$ and the pmf $p_c(\cdot)$ are coupled. Therefore, we first solve (\ref{main-opt}) by optimizing the pmf $p_c(\cdot)$ and then, determine the $\gamma_1$ value that satisfies (\ref{optimalgammaeq}).

We now investigate different special network scenarios where significant simplification is possible.

\subsection{Rayleigh Fading $\dsr$ Results}\label{Rayleighresults}
We optimize the $\dsr$ of users for the multi-file setup, where interference fading power follows an exponential distribution with $g \sim \exp(\mu)$. We consider several special cases corresponding to 1) small but non-zero noise, 2) arbitrary noise and 3) an approximation for arbitrary noise allowing the path loss exponent $\alpha = 4$. We find the optimal caching distribution corresponding to each scenario.


\begin{lem}\label{smallnoiselemma}
{\bf Small but non-zero noise, $\alpha>2$.} The optimal caching distribution is $p_c(i)={\frac{1}{i^{\gamma_c}}}\Big/{\sum\limits_{j=1}^{M}{\frac{1}{j^{\gamma_c}}}}, \, i=1,\hdots, M$, which is also Zipf distributed, where $\gamma_c=\frac{\gamma_r}{\alpha/2+1}$ is the Zipf exponent for the  caching pmf.
\end{lem}

\begin{proof}
See Appendix \ref{App:AppendixB}.
\qedhere
\end{proof}

Assuming $\alpha>2$, the caching pmf exponent satisfies $\gamma_c<\frac{\gamma_r}{2}$, which implies that the optimal caching pmf that maximizes the $\dsr$ has a more uniform distribution exhibiting less locality of reference compared to the request distribution that is more skewed towards the most popular files.

\begin{assu}\label{assumptionbeta}
{\bf Separability of coverage distribution.} For Rayleigh, Ricean and Nakagami small-scale fading distributions, the function $\beta({\T},\alpha)^{\alpha/2}$ can be approximated as a linear function of $\T$ as shown in Fig. \ref{betavsT}. This relation\footnote{Although the expression $\beta({\T},\alpha)^{\alpha/2}/\T$ is not analytically tractable, we can approximate $\beta({\T},\alpha)^{\alpha/2}$ as a linear function of $\T$ because the lower incomplete Gamma function has light-tailed characteristics. Since the channel power distribution -which is exponential due to Rayleigh fading- is also light tailed, we can expect to observe such a linear approximation in our numerical results.} greatly simplifies the analysis of the optimization problem given in (\ref{main-opt}). 
\end{assu}

\begin{lem}\label{ArbitraryNoise}
{\bf Arbitrary Noise, $\alpha>2$.} For arbitrary noise, from Assumption \ref{assumptionbeta}, the optimal caching distribution $p_c(\cdot)$ can be approximated as a Zipf distribution given by
\begin{eqnarray}
\label{arbitrarynoisepmf}
p_c(i)\approx{\frac{1}{i^{\gamma_c}}}\Big/{\sum\limits_{j=1}^{M}{\frac{1}{j^{\gamma_c}}}}, \quad i=1,\hdots, M,
\end{eqnarray}
where $\gamma_c=\frac{\gamma_r}{\alpha/2+1}<\frac{\gamma_r}{2}$ is the Zipf exponent for the  caching pmf assuming $\alpha>2$.
\end{lem}

\begin{proof}
See Appendix \ref{App:AppendixC}.
\end{proof}

Interestingly, this result is the same as Rayleigh fading with small but non-zero noise model developed in Sect. \ref{Rayleighresults}, which follows from the monotonic transformation \cite{Carter2001} caused by increasing the noise power $\sigma^2$ in (\ref{thinning}). According to the pmf given in (\ref{arbitrarynoisepmf}), the optimal caching strategy exhibits less locality of reference than the input stream of requests. Therefore, it is a good caching strategy, which will be further verified in Sect. \ref{Performance}. Lemma \ref{ArbitraryNoise} suggests that files with higher popularity should be cached less frequently than the demand for this file, and unpopular files should be cached more frequently than the demand for the file. However, high popularity files should be still cached at more locations compared to the low popularity files. The path loss evens out the file popularities and the caching distribution should be more uniform compared to the request distribution. The sequential transmission model shows that for a Zipf request distribution with exponent $\gamma_r$, which is skewed towards the most popular files, the optimal caching pmf should be also Zipf distributed with the relation $\gamma_c<\frac{\gamma_r}{2}$ for $\alpha>2$, implying that the caching pmf is more uniform than the request pmf.

The next result generalizes Lemma \ref{ArbitraryNoise} to any request distribution $p_r(\cdot)$ rather than the Zipf distribution, and is derived solving (\ref{derivative-relation}) in Appendix \ref{App:AppendixC} using the separability of coverage from Assumption \ref{assumptionbeta}.
\begin{theo}\label{GeneralFadingArbitraryNoiseArbitraryRequest}
For arbitrary noise, if the small-scale fading is Rayleigh, Nakagami or Ricean distributed, from Assumption \ref{assumptionbeta}, for a general request pmf, $p_r(\cdot)$, the optimal caching pmf is approximated as 
\begin{eqnarray}
\label{generaloptimal}
p_c(i)\approx {p_r(i)^{\frac{1}{(\alpha/2+1)}}}\Big/{\sum\limits_{j=1}^{M}{p_r(j)^{\frac{1}{(\alpha/2+1)}}}}, \quad i=1,\hdots, M.
\end{eqnarray}
\end{theo}
From (\ref{generaloptimal}), it is required to flatten the request pmf to optimize the caching performance. Examples include the case of uniform demands, where the optimal caching distribution should be also uniform, and Geometric($p$) request distribution, for which the caching distribution satisfies Geometric($q$), where $q=1-(1-p)^{\frac{1}{(\alpha/2+1)}}$. In the case of Zipf demands, we can derive the same result as in Lemma \ref{ArbitraryNoise}. 


\begin{figure}[t!]
\centering
\begin{minipage}{.328\textwidth}
\centering
\includegraphics[width=\columnwidth]{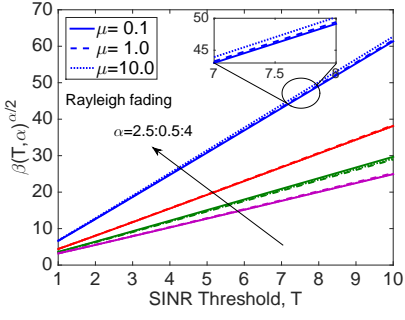}
\center{\small{(a) Rayleigh fading, param. $\mu$.}}
\end{minipage}
\begin{minipage}{.328\textwidth}
\centering
\includegraphics[width=\columnwidth]{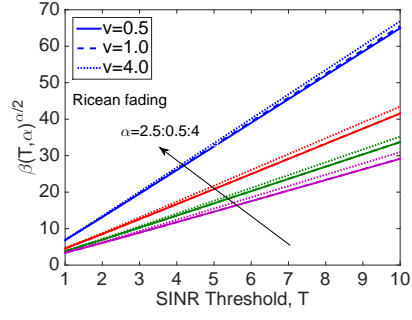}
\center{\small{(b) Ricean fading, distance param. $v$. 
}}
\end{minipage}
\begin{minipage}{.328\textwidth}
\centering
\includegraphics[width=\columnwidth]{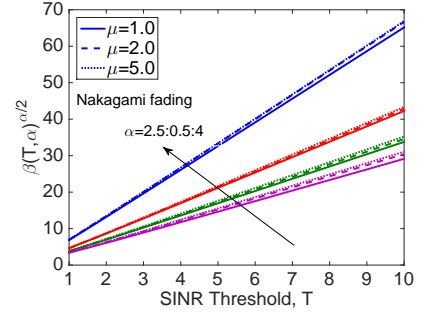}
\center{\small{(c) Nakagami fading, shape param. $\mu$.
}}
\end{minipage}
\caption{\small{The linear relation between $\beta({\T},\alpha)^{\alpha/2}$ and $\T$.}\label{betavsT}}
\end{figure}

\begin{lem}\label{ArbitraryNoiseApprox}
{\bf An Approximation for Arbitrary Noise with $\alpha=4$.} For a total number of files $M$ and arbitrary noise with $\alpha=4$, the optimal caching pmf is 
\begin{eqnarray}
\label{arbitNoiseApprox}
p_c(i)=a_i+b\log\Big(\frac{i+1}{i}\Big),\,\, i=1, \hdots, M,
\end{eqnarray}
where $b=\frac{\sqrt{\mu \T \sigma^2}\gamma_r}{\pi \lambda_t\beta({\T},4)}$, $a_i=\frac{1}{M}+\frac{b}{M}\sum\limits_{j=1}^M{\log\left(\frac{j}{i+1}\right)}$, and the pmf is valid only if $b\leq [M\log(M)-\log(M!)]^{-1}$. 
\end{lem}

\begin{proof}
See Appendix \ref{App:AppendixD}.
\end{proof}

The distribution $p_c(\cdot)$ in (\ref{arbitNoiseApprox}) of Lemma \ref{ArbitraryNoiseApprox}  is a variety of Benford's law \cite{Pietronero2001}, which is a special bounded case of Zipf's law. Benford's law refers to the frequency distribution of digits in many real-life sources of data and is characterized by the pmf $F_B(i)=\log_{10}\left(\frac{i+1}{i}\right),\,\, i\in\{1, \hdots, 9\}$. In distributed caching problems, the number of files, $M$, is generally much greater than 9. Therefore, we generalize the law as $F_B(i)=\log_{M+1}\left(\frac{i+1}{i}\right)$, $i\in \{1, \hdots, M\}$. The result in (\ref{arbitNoiseApprox}) has a very similar form as the Benford law with shift parameter $a_i$ for file $i$ and a scaling parameter $b$, as determined in Lemma \ref{ArbitraryNoiseApprox}.

\section{A Lower and Upper Bound on the $\dsr$ and Different Caching Strategies}
\label{bounds}
The analysis of the $\dsr$ becomes intractable for the multiple file case when the caching pdf does not have a simple form. Therefore, we derive a lower and upper bound to characterize the $\dsr$ for the  sequential serving model and provide two different caching strategies to maximize $\DSRs$.  

\subsection{Bounds on $\DSRs$}
\label{boundsSequential} 
We provide a lower and upper bound for $\DSRs$, the $\dsr$ of the sequential serving-based transmission model with multiple files. We discussed the optimal file caching problem for multiple file scenarios in \cite{Malak2014}. Here, we compare our solution to the several bounds and other caching strategies.

\subsubsection{Upper Bound (UB)} Using the concavity of $\pcov({\T},\lambda_t p_c(i),\alpha)$ in $p_c(i)$, a UB is found as
\begin{eqnarray}
\sum\nolimits_{i=1}^{M}{p_r(i)\pcov({\T},\lambda_t p_c(i),\alpha)}\,\substack{(a)\\<}\,\pcov\Big({\T},\lambda_t \sum\nolimits_{i=1}^{M}{p_r(i)p_c(i)},\alpha\Big)
\,\substack{(b)\\ \leq}\, \pcov({\T},\lambda_t p_r(1),\alpha),
\end{eqnarray}
where ($a$) follows from Jensen's inequality, and ($b$) follows from the assumption $p_r(1)>p_r(i)$ for $1<i\leq M$ that yields $\sum\nolimits_{i=1}^{M}{p_r(i)p_c(i)}<p_r(1)\sum\nolimits_{i=1}^{M}{p_c(i)}=p_r(1)$, where $p_r(1)=\Big(\sum\nolimits_{j=1}^{M}{j^{-\gamma_r}}\Big)^{-1}$.

\subsubsection{Lower Bound (LB)} 
Using the fact that given $p_r(\cdot)$ is Zipf distributed, the optimal $p_c(\cdot)$ also has Zipf distribution as proven in Lemma \ref{ArbitraryNoise} as a solution of the $\DSRs$ maximization problem in (\ref{main-opt}). As a result, any distribution that is not skewed towards the most popular files will yield a suboptimal $\DSRs$. Hence a uniform caching distribution performs worse than the Zipf law, and a LB is found as
\begin{eqnarray}
\label{LB}
\sum\nolimits_{i=1}^{M}{p_r(i)\pcov({\T},\lambda_t p_c(i),\alpha)}>\sum\nolimits_{i=1}^{M}{p_r(i)\pcov\Big({\T},\frac{\lambda_t}{M},\alpha\Big)}=\pcov\Big({\T},\frac{\lambda_t}{M},\alpha\Big).
\end{eqnarray}

\subsection{Caching Strategies for the Sequential Serving Model with Multiple Files}
\label{opt-multi-file}
We propose two optimization formulations to maximize $\DSRs$ in the presence of multiple files, where the request and caching probabilities are known a priori because in general the optimal $\DSRs$ and the optimal caching distribution is not tractable. The first strategy, where we maximize the $\dsr$ for the least popular file, favors the least desired file, i.e., the file with the lowest popularity, to prevent from fading away in the network. Therefore, we introduce the variables $0\leq \rho_i \leq 1$ for files $i\in\{1,\cdots, M\}$ to weight the caching pmf $p_c(\cdot)$. The second strategy aims to maximize the $\dsr$ of all files by optimizing the fraction $\rho_i$'s of the users for each file type. We assume the caching distribution is given. Then, we provide iterative techniques to solve the problems presented in this section.

\subsubsection{Maximum $\dsr$ of the Least Desired File} 
\label{mul-file-opt}
Our motivation behind maximizing the $\dsr$ of the least desired file is to prevent the files with low popularity from fading away in the network. 

\begin{lem}
\label{mul-file-optlemma}
The caching probability of each file is weighted by $\rho_i<1$ so that the total fraction of transmissions for all files, denoted by $\xi$ satisfies $\xi=\sum\nolimits_{i=1}^{M}{\rho_i p_c(i)}\leq 1$. Given $\eta=\underset{i,\,\,\rho_i=1}{\max}\,\, {p_r(i)p_c(i)}=p_r(j)p_c(j)$ for some $j$, the optimal solution is given by $\rho_i=1_{\{i\geq j\}}+\frac{\eta}{p_r(i)p_c(i)} 1_{\{1\leq i < j\}}$.
\end{lem}

\begin{proof}
See Appendix \ref{AppendixMINdsr}.
\end{proof}

\subsubsection{Maximum $\dsr$ of All Files}
\label{multifileopt2}
We maximize the $\dsr$ for all files without any prioritization.
\begin{lem}
\label{multifileopt2lemma}
The optimal solution to maximize the $\dsr$ for all files is given by
$\rho_i=1$ for all $i$.
\end{lem}

\begin{proof}
See Appendix \ref{AppendixMAXdsr}.
\end{proof}

As well as maximizing the $\dsr$ for the sequential model, one might wish to select a file with a particular request probability, and use $\DD$ to distribute this file and all files with higher probability or simultaneously cache all files using $\DD$ as detailed in Sect. \ref{mostgeneralcase}. In the next section, we describe the simultaneous transmission of multiple files, and derive expressions for $\SINR$ distribution and $\dsr$.

\section{Simultaneous Transmissions of Different Files with Arbitrary Noise}
\label{mostgeneralcase}
We consider the multiple file case, where a typical receiver requires a specific set of files, and the set of its transmitter candidates are the ones that contain any of the requested files. Each receiver gets the file from the closest transmitter candidate. The rest of the active transmitters that do not have the files requested are the interferers. We provide a detailed analysis for the $\SINR$ coverage next.

Assume that each receiver has a state, determined by the set of files it requests. For a receiver in state $j$, the set of requested files is $f_r(j)$. Let the tagged receiver be $y\in\Phi_r$ and in state $j$, and $\Phi_{t}(j)$ be the set of transmitters that a receiver in state $j$ can get data from. Hence, the set of transmitter candidates for user in state $j$ is the superposition given by $\Phi_{t}(j)=\sum_{i\in f_r(j)}\Phi_{t,i}$, where $\Phi_{t,i}$ is the set of transmitters containing file $i$. Let $\lambda_j$ be the density of $\Phi_{t}(j)$, where $\lambda_j=\lambda_tp_j=\lambda\gamma_1p_j$. The rest of the transmitters, i.e., $\sum_{i\notin f_r(j)}\Phi_{t,i}$, is an independent process with density $\lambda_t-\lambda_j=\lambda_t(1-p_j)=\lambda\gamma_1(1-p_j)$.

The sum $p_j=\sum\nolimits_{i\in f_r(j)}p_c(i)$ gives the probability that the user has at least one of the files requested by any receiver in state $j$. Hence, the density of the transmitter candidates $\lambda_j$ for a receiver in state $j$ are given by the product of $\lambda\gamma_1$ and $\sum\nolimits_{i\in f_r(j)}p_c(i)$, i.e., $\lambda_j=\lambda_tp_j=\lambda\gamma_1\sum\nolimits_{i\in f_r(j)}p_c(i)$. Hence, using the nearest neighbor distribution of the typical receiver in state $j$, the distance to its nearest transmitter is distributed as ${\Rayleigh}(\sigma_j)\sim \frac{r}{\sigma_j^2}\exp\Big({-\frac{r^2}{2\sigma_j^2}}\Big)$, for $\sigma_j=1/\sqrt{2\pi\lambda_j}$ and $r\geq 0$. 

We assume that all users experience Rayleigh fading with mean $1$, and constant transmit power of $1/\mu$. Assuming user $y$ is at o, in state $j$ and is a receiver, and x is the tagged transmitter denoted by $b_o$, and the distance between them is $r$, then the $\SINR$ at user $y$ is ${\SINR}_j = \frac{hr^{-\alpha}}{\sigma^2+I_{r(j)}}$, where $h$ is the channel gain parameter between $x$ and $y$, $\sigma^2$ is the white Gaussian noise, and $I_{r(j)}$ is the total interference at node $y$ in state $j$, and given by the following expression: $I_{r(j)}=\sum\nolimits_{z\in \Phi_{t}\backslash b_o}{g_{z}r_z^{-\alpha}}=\sum\nolimits_{z\in \Phi_{t}(j)\backslash b_o}{g_zr_z^{-\alpha}}+\sum\nolimits_{z\in \Phi_{t}\backslash \Phi_t(j)}{g_zr_z^{-\alpha}}$, where $g_{z}$ is the channel gain from the interferer $z$ and the receiver $y$, $r_z$ is the interferer $z$ to receiver distance, on $\RHS$, the first term is the interference due to the set of transmitters that has the files requested by the receiver, and the second term is the interference due to the rest of the transmitters that do not have any of the desired files by the receiver. The total interference depends on the transmission scheme. Compared to the nearest user association \cite{Andrews2011}, it is hard to characterize the interference in dynamic caching models with different association techniques.

\begin{theo}
\label{mainProbCov}
The probability of coverage of a typical user conditioned on being at state $j$ is given by\footnote{The definition of $\PCOV({\T},\lambda_j,\alpha)$ here is different from the definition of the classical downlink coverage probability $\pcov({\T},\lambda,\alpha)$ given in (\ref{pcov}) due to the possibility of simultaneous transmissions of different file types.}
\begin{eqnarray}
\label{maincoverageprobability}
{\PCOV({\T},\lambda_j,\alpha)}=\pi\lambda_j \int\nolimits_{0}^{\infty}{ e^{-\pi\lambda_jv(1-\rho_2({\T},\alpha))-\pi\lambda_t v(\rho_1({\T},\alpha)+\rho_2({\T},\alpha))-{\T}\sigma^2v^{\alpha/2}}\, \mathrm{d}v},
\end{eqnarray}
where $\rho_1(\T,\alpha)={\T}^{2/\alpha} \int\nolimits_{{\T}^{-2/\alpha}}^{\infty}{ \frac{1}{1+u^{\alpha/2}}\, \mathrm{d}u}$ and $\rho_2(\T,\alpha)={\T}^{2/\alpha} \int\nolimits_{0}^{{\T}^{-2/\alpha}}{ \frac{1}{1+u^{\alpha/2}}\, \mathrm{d}u}$.
\end{theo}

\begin{proof}
See Appendix \ref{App:AppendixA0}.
\qedhere
\end{proof}

We now consider the special case of the path loss exponent $\alpha=4$, which is more tractable. 

\begin{cor} 
\label{ProbCov0}
Letting $H({\T},\lambda_t,p_j)=\Big(\frac{p_j}{\sqrt{\T}}-p_j\tan^{-1}\left(\frac{1}{\sqrt{\T}}\right)+\frac{\pi}{2}\Big)\frac{\pi \lambda_t}{\sqrt{2\sigma^2}}$, the probability of coverage of a typical user conditioned on being at state $j$ for the special case of $\alpha=4$ and $\mu=1$ is given by
\begin{eqnarray}
\label{coverageprobability}
{\PCOV({\T},\lambda_j,4)}=\pi\lambda_tp_j\sqrt{\frac{\pi}{\T\sigma^2}}e^{\frac{H({\T},\lambda_t, p_j)^2}{2}}{\Q} \big(H({\T},\lambda_t,p_j)\big).
\end{eqnarray}
\end{cor}

\begin{proof}
See Appendix \ref{App:AppendixB0}.
\qedhere
\end{proof}
  
Since the term $\sqrt{\T}\tan^{-1}\big(\frac{1}{\sqrt{\T}}\big)$ is increasing in $\T$ and converges to $1$ in the limit as $\T$ goes to infinity, $H({\T},\lambda_t,\cdot)$ is increasing in $p_j$, and positive. Furthermore, $\PCOV({\T},\lambda_j,\alpha)$ is monotonically increasing in $p_j$. This observation is essential in the characterization of the $\dsr$ under different user criteria.

We consider two different strategies for the simultaneous transmission of multiple files, namely popularity-based and global models, which differ mainly in the set of files cached at the transmitters.

\subsection{Popularity-based $\dsr$}
\label{PUF}
In this approach, a set of files corresponding to the most popular ones in the network is cached simultaneously at all transmitters. We define $\DSRp$, which stands for the $\dsr$ of the popularity-based approach, and is calculated over the set of most popular files as 
\begin{eqnarray}
\label{UP}
{\DSRp}=\lambda\gamma_2 \sum\nolimits_{k\in\mathcal{K}}p_r(k)\PCOV({\T},\xi_l,\alpha),
\end{eqnarray}
where $\mathcal{K}$ is the set of the $K$ most popular files, and $\xi_l=\lambda\gamma_1 \sum\nolimits_{i\in\mathcal{L}}p_c(i)$, where $\mathcal{L}$ is a set corresponding to the most popular $K$ files cached at the transmitters among the set of available files in the caches.

Consider the special case of (\ref{UP}), where only the most popular file in the network is cached at all the transmitters if available, i.e., $|\mathcal{K}|=1$, which modifies (\ref{UP}) as ${\DSRp}=\lambda\gamma_2 p_r(k)\PCOV({\T},\lambda\gamma_1 p_c(k),\alpha)\stackrel{(a)}{=}\lambda\gamma_2 p_r(k)\pcov({\T},\lambda\gamma_1 p_c(k),\alpha)$, where $(a)$ follows from the fact that for $|\mathcal{K}|=1$, the coverage probability becomes same as the sequential serving-based model in Sect. \ref{MultipleFiles}, and the most popular file index $k$ can be found from the demand distribution and is given by $k=\underset{i\in \{1,\hdots,M\}}{\argmax}\hspace{0.2cm} p_r(i)$, and hence the corresponding density of the transmitters is $\lambda\gamma_1 p_c(k)$, where $p_r(k)\geq p_r(l)$ for all $l=1,\hdots, M$.

\subsection{Global $\dsr$}
\label{GUF}
Global $\dsr$ is defined as the average performance of all users in the network, which is determined by the spatial characteristics of file distributions and the coverage of a typical user. The $\dsr$ function in our model is state dependent since the coverage probability of a user is determined according to the files requested by the user. The expected global $\dsr$ is given as follows:
\begin{eqnarray}
\label{expectedutility}
{\DSRg}=\lambda\gamma_2\sum\nolimits_{i=1}^M{p_r(i)\PCOV({\T},\gamma_1\lambda p_c(i),\alpha)}.
\end{eqnarray}

{\bf A Discussion on the Various Transmission Models.}
Popularity-based transmission and global model in this section do not depend on the cache states. Instead, they both depend on the global file popularity distributions, and have similar characteristics as given in (\ref{UP}) and (\ref{expectedutility}). It is intuitive to observe that the optimal caching distributions in both models follow similar trends as the request distribution. Sequential serving-based model in Sect. \ref{SUF} boils down to the scenario characterized in \cite{Andrews2011} where only a subset of transmitters and their candidate receivers are active simultaneously. Hence, this model mitigates interference and provides higher coverage than the other models. However, since the $\dsr$ is a weighted function of the file transmit pmf $p_c(\cdot)$, the $\dsr$ of the model is reduced.  

\begin{figure}[t!]
\centering
\begin{minipage}{.49\textwidth}
  \centering
  \includegraphics[width=0.8\linewidth]{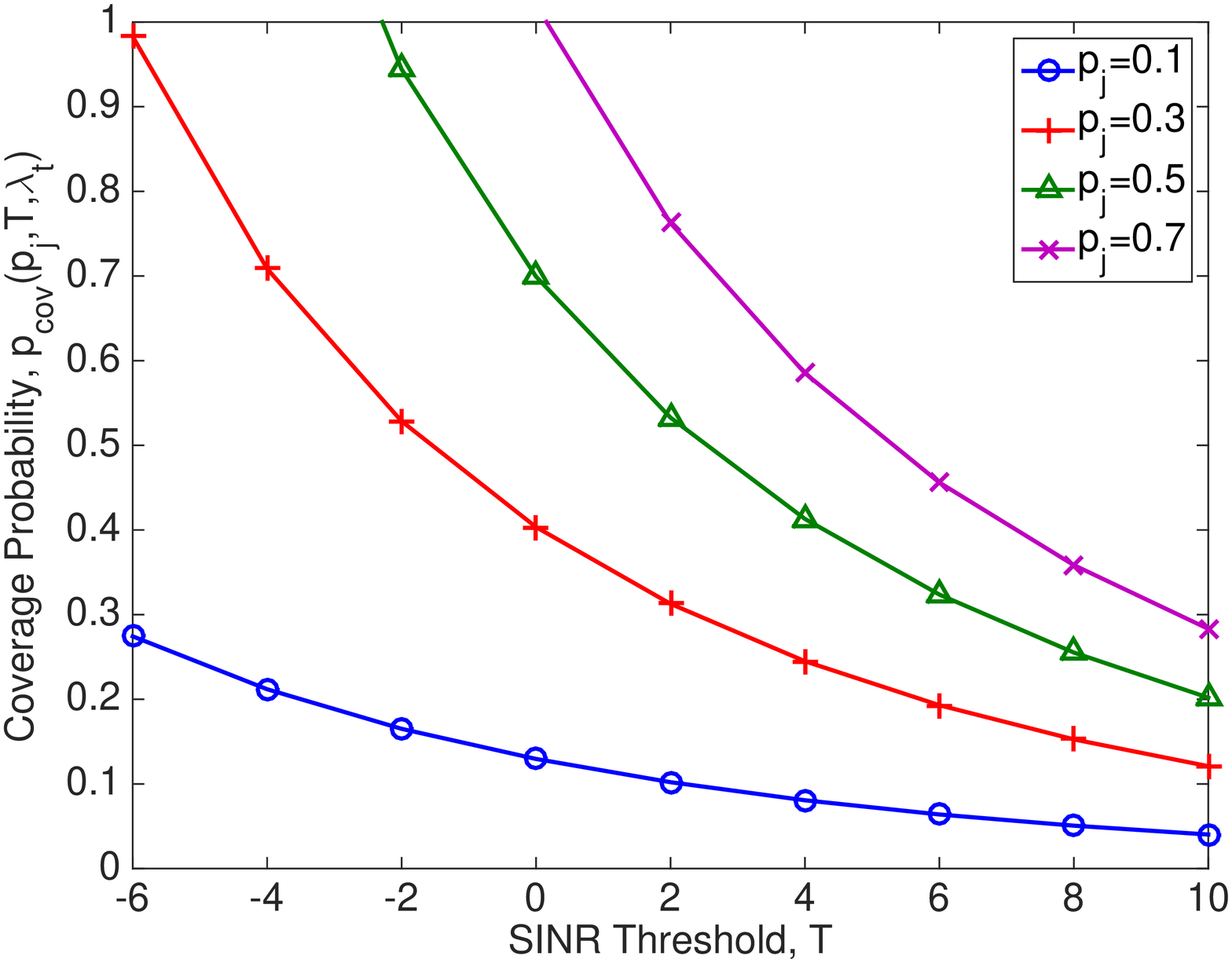}
  \caption{\small{Analytical model for the $\SINR$ coverage probability for different transmitter densities where $\lambda=1$ and $\gamma_1=0.4$.}}
  \label{SINRcoverageprob}
\end{minipage}
\hfill
\begin{minipage}{.49\textwidth}
  \centering
  \includegraphics[width=0.8\linewidth]{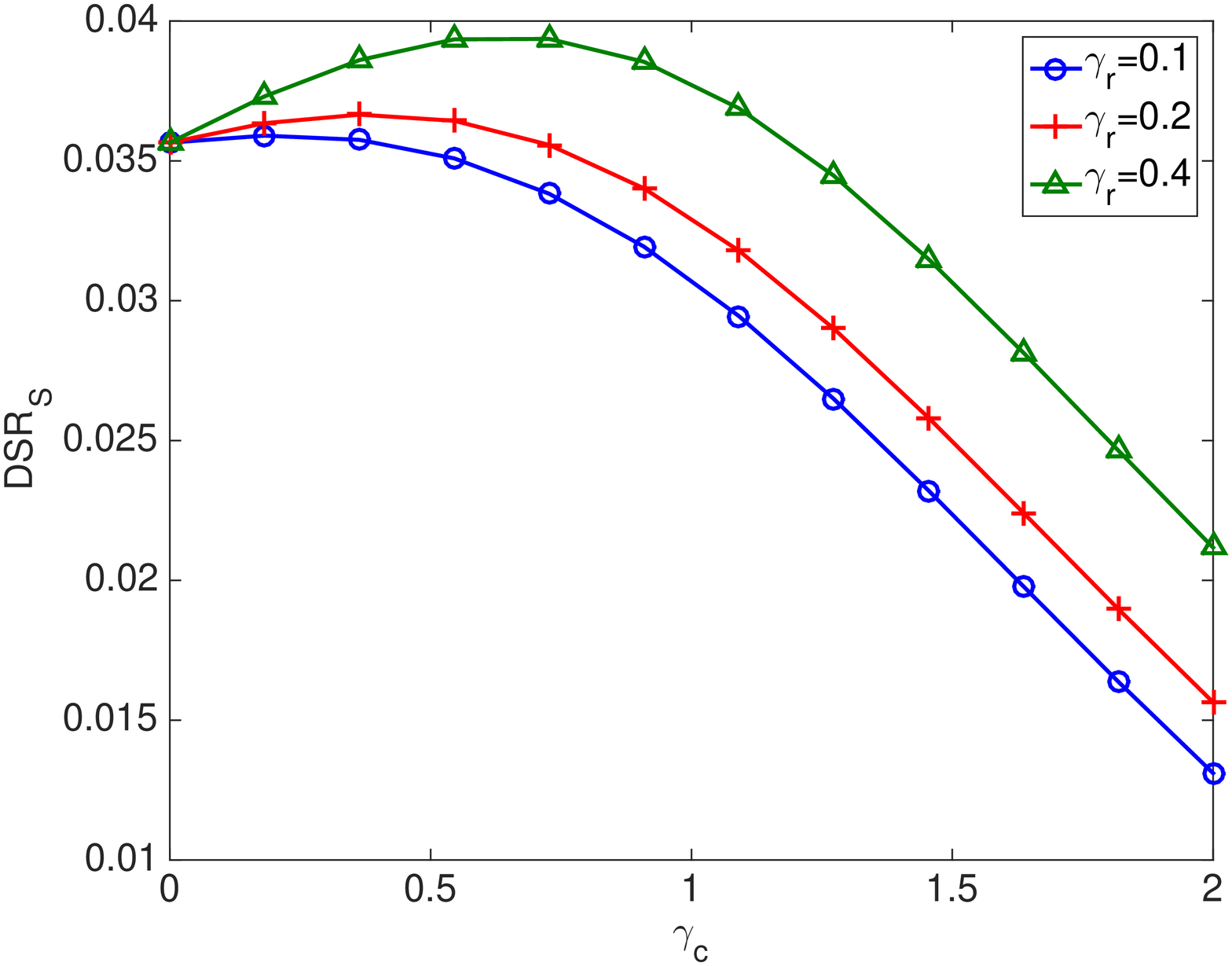}
  \caption{\small{Average $\dsr$ for the sequential model, $\DSRs$ versus $\gamma_c$ for Zipf request and Zipf caching distributions.}}
  \label{sequtility}
\end{minipage}
\end{figure}

\begin{figure}[t!]
\centering
\begin{minipage}{.49\textwidth}
  \centering
  \includegraphics[width=0.8\linewidth]{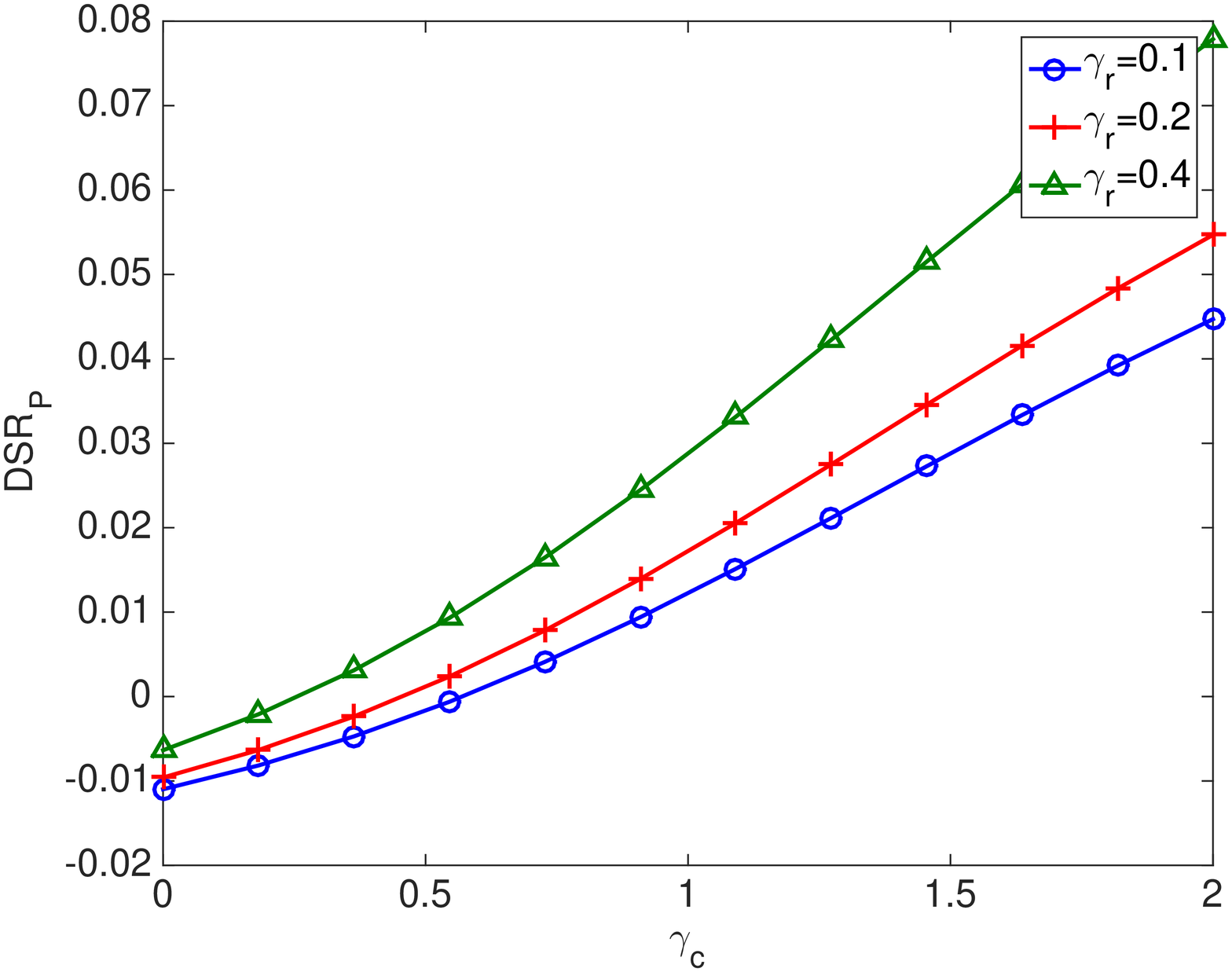}
  \caption{\small{Average $\dsr$ for the popularity-based model, $\DSRp$ versus $\gamma_c$ for Zipf request and Zipf caching distributions.}}
  \label{poputility}
\end{minipage}
\hfill
\begin{minipage}{.49\textwidth}
  \centering
  \includegraphics[width=0.8\linewidth]{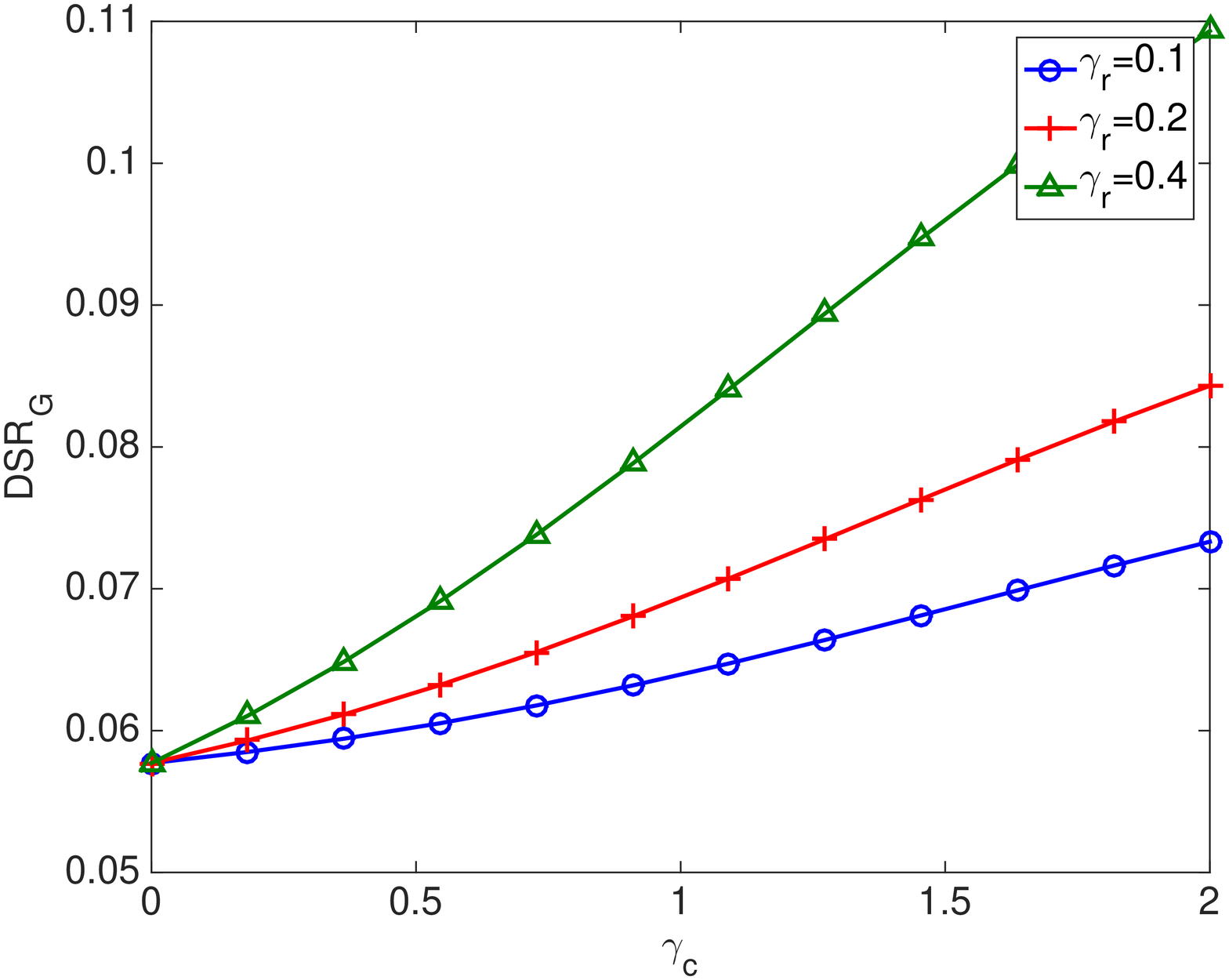}
  \caption{\small{Average $\dsr$ for the global model, $\DSRg$ versus $\gamma_c$ for Zipf request and Zipf caching distributions.}}
  \label{globalutility}
\end{minipage}
\end{figure}

Now, we present some numerical results on the general transmission models discussed and present results related to the popularity-based $\dsr$, global $\dsr$ and sequential $\dsr$.

{\bf State dependent coverage probability.} We illustrate the $\SINR$ coverage probability for varying $p_j$ for a fixed fraction of transmitters ($\gamma_1=0.4$) in Fig. \ref{SINRcoverageprob}. The coverage probability is state dependent\footnote{The receiver's state refers to the collection of files it requests.} and for the receiver in state $j$, the density of transmitters is given by $\lambda_j=\lambda p_j$ where $p_j=\gamma_1\sum\limits_{i\in f_r(j)}p_c(i)$. If the requested files are available in the set of transmitters, then the receiver has higher coverage. Therefore, for higher fraction of transmitters $\gamma_1$, the coverage probability is higher.

{\bf Caching performance of the proposed transmission models.} The optimal caching strategies that maximize the caching problems of Sect. \ref{mostgeneralcase} given in (\ref{UP}) and (\ref{expectedutility}) are not necessarily Zipf distributed. However, without the Zipf distribution assumption, the optimization formulations become intractable since $\pcov({\T},\lambda_j,\alpha)$ in (\ref{maincoverageprobability}) is nonlinear in the density of the users. Therefore, for simulation purposes, we find the optimal Zipf caching exponents that maximize the proposed functions.

{\bf $\dsr$ comparison.} We investigate the variation of the sequential model $\DSRs$ with respect to the caching parameter $\gamma_c$. From Fig. \ref{sequtility}, we observe that $\gamma_c$ increases with the request distribution parameter $\gamma_r$, assuming both distributions are Zipf. In Figs. \ref{poputility} and \ref{globalutility}, we illustrate the variation of the popularity-based model $\DSRp$ and the global model $\DSRg$ with $\gamma_c$. In both figures, it is clearly seen that as the requests become more skewed (higher $\gamma_r$), the $\dsr$ increases. It also increases with $\gamma_c$, which implies that the optimal caching distribution should also be skewed towards the highly popular files.

\section{Numerical Results and Discussion}
\label{Performance}
We evaluate the optimal caching distributions that maximize the $\dsr$. The simulation results are based on Sects. \ref{MultipleFiles} and \ref{bounds}. We consider a general $\PPP$ network model with Rayleigh fading distribution with $\mu=1$ and $\alpha=4$ for small and general noise solutions. The requests are modeled by $\textrm{Zipf}(\gamma_r)$. 

\begin{figure}[t!]
\centering
\begin{minipage}{.49\textwidth}
\centering
\includegraphics[width=0.8\linewidth]{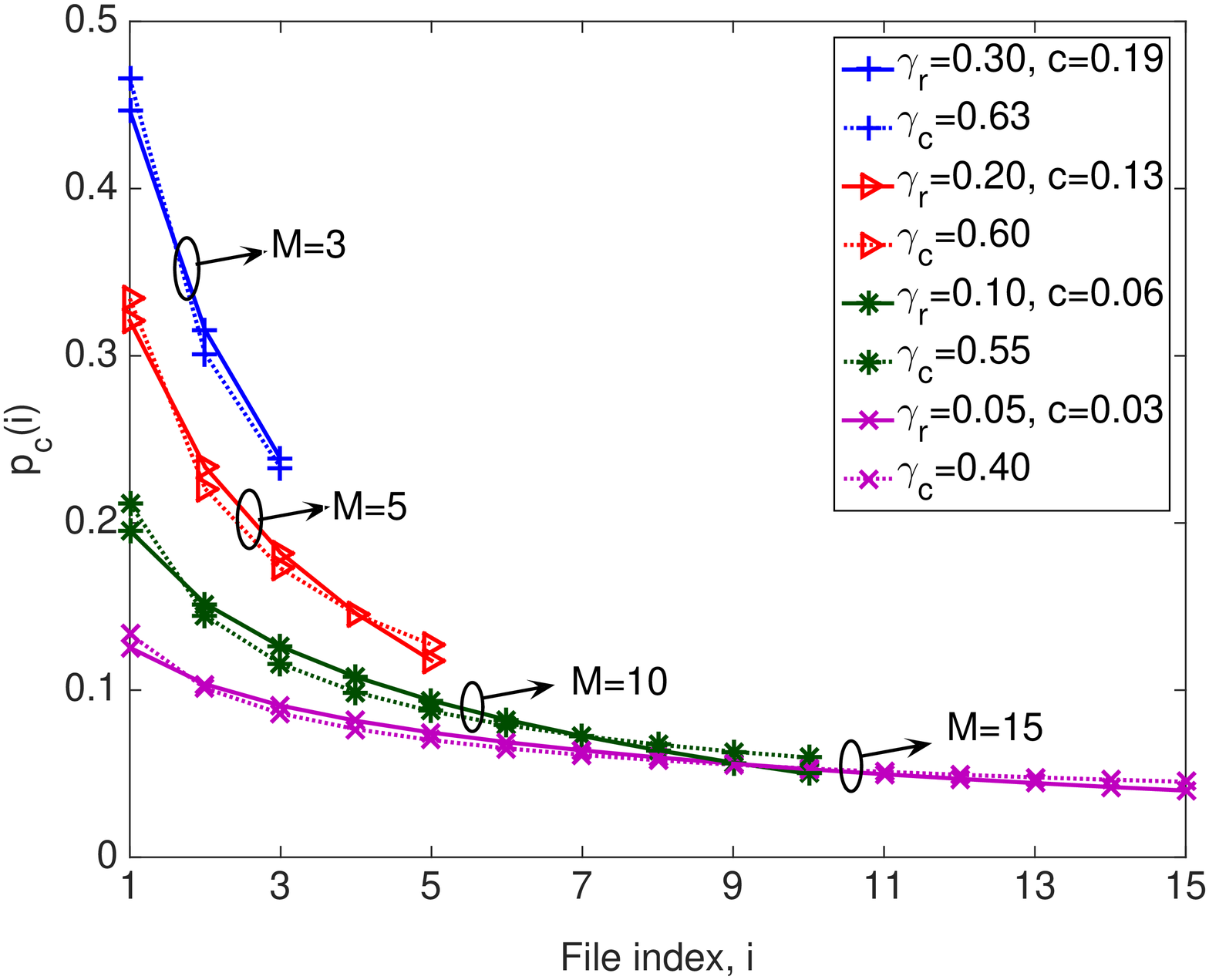}
\caption{\small{For a $\textrm{Zipf}(\gamma_r)$ popularity distribution, Benford law and approximate $\textrm{Zipf}(\gamma_c)$ caching pmf for various $M$.}} 
\label{benfordfig}
\end{minipage}
\hfill
\begin{minipage}{.49\textwidth}
\centering
\includegraphics[width=0.8\linewidth]{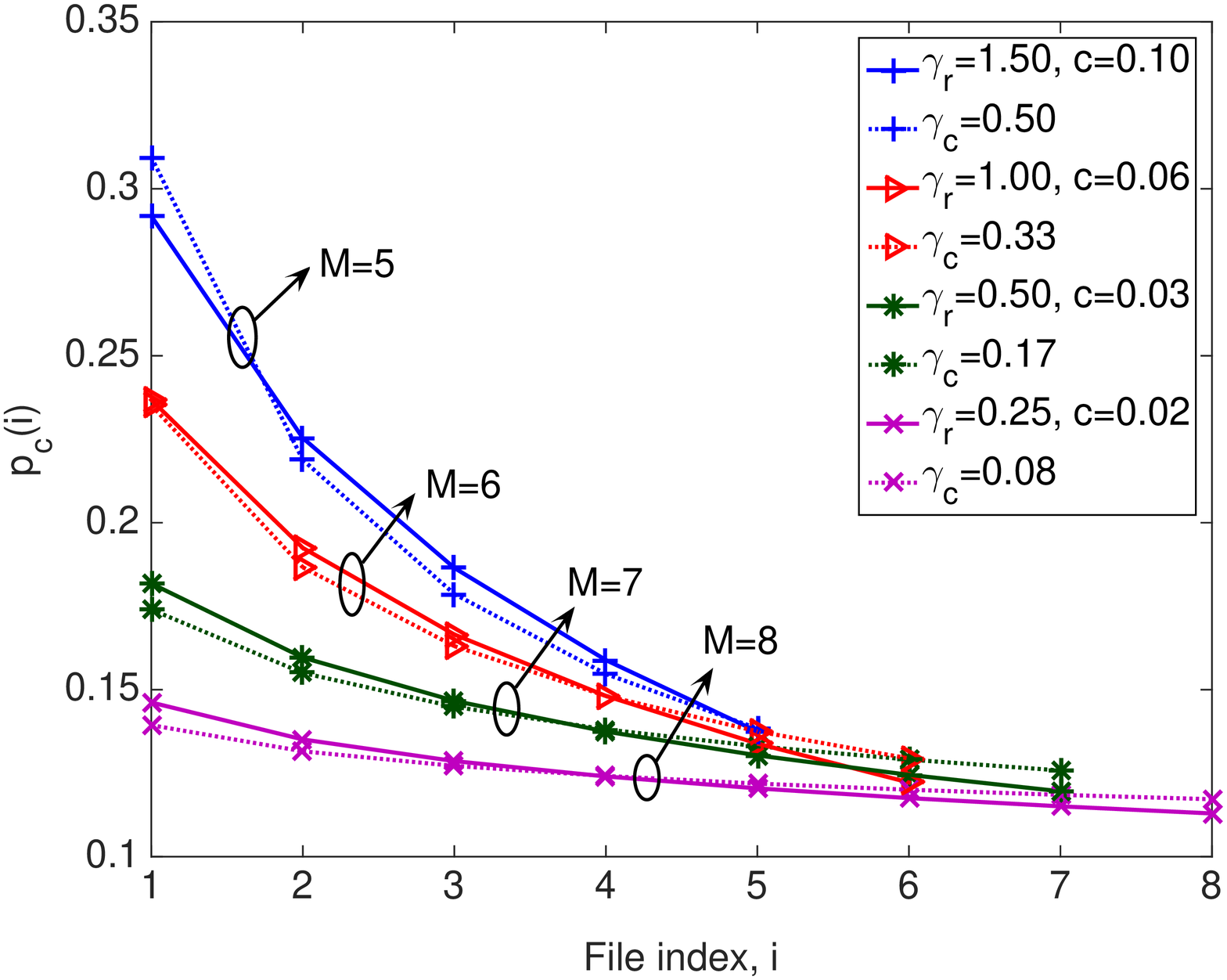}
\caption{\small{For a $\textrm{Zipf}(\gamma_r)$ popularity distribution, Benford law and optimal $\textrm{Zipf}(\gamma_c)$ pmf for different $M$ ($\SNR=30$ dB).}}
\label{benfordZipf} 
\end{minipage}
\end{figure}

{\bf Benford versus Zipf distributions.}
In Figs. \ref{benfordfig} and \ref{benfordZipf}, we illustrate the trend of optimal Zipf caching distribution and the Benford law developed in Sect. \ref{MultipleFiles} for different numbers of total files. As seen from Fig. \ref{benfordfig}, these two distributions have similar characteristics. However, as $\gamma_r$ increases, the range of $M$ for which Benford caching distribution in (\ref{arbitNoiseApprox}) and Zipf laws are comparable becomes narrower. For $\gamma_r>0.3$, it is not practical to approximate the Benford law with a Zipf distribution. In fact, as described in Sect. \ref{MultipleFiles}, as the noise level decreases, i.e., $b={\sqrt{\mu \T \sigma^2}\gamma_r}/{(\pi \lambda_t\beta({\T},4))}$ drops, the optimal caching strategy converges to Zipf distribution. As seen in Fig. \ref{benfordZipf}, for small noise, i.e., for high $\SNR$, these laws behave similarly for relatively high $\gamma_r$ values compared to the general noise case. 

We now compare the $\dsr$ of the sequential serving model for various $\gamma_r$ based on the optimal solutions that are also Zipf distributed, as derived in Sect. \ref{MultipleFiles}, and the lower and upper bounds obtained in Sect. \ref{bounds}. The numerical solutions are obtained by calculating the $\dsr$ of various (random) caching distributions and picking the best one that achieves the highest $\dsr$.

{\bf Zipf caching with $\gamma_c=\frac{\gamma_r}{(\alpha/2+1)}$ is a good approximation to maximize the $\dsr$.} In Fig. \ref{05SNR1new}, we compare the performances of different caching strategies for a Zipf request distribution with parameter $\gamma_r=0.5$ and $\SNR=1$. The Zipf caching distribution with parameter $\gamma_r/3$ is very close to the optimal solution evaluated numerically that is also very close to the simple lower bound derived in (\ref{LB}). Furthermore, Benford distribution has very similar characteristics as the optimal caching distribution solution. There is a huge gap between the UB and the no noise in terms of the $\dsr$, and the $\dsr$ for the no noise case is the highest among all for all $\SNR$ or $\T$ values.

\begin{figure}[t!]
\centering
\begin{minipage}{.49\textwidth}
\centering
\includegraphics[width=0.8\linewidth]{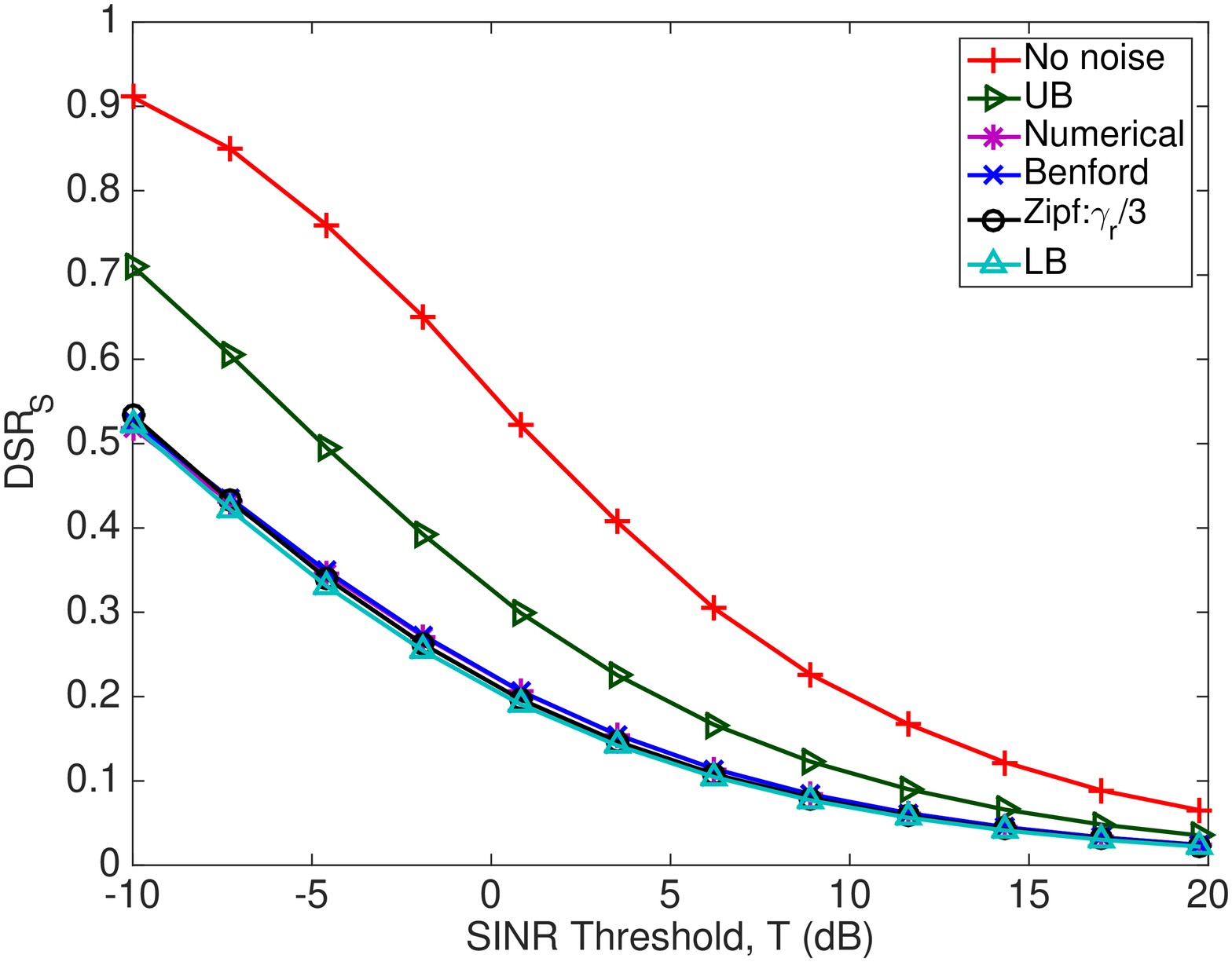}
\caption{\small{Bounds and approximations to the optimal $\DSRs$ for $M=10$, $\SNR=1, \lambda=1$, Zipf request pmf with $\gamma_r=0.5$.}}
\label{05SNR1new}
\end{minipage}
\hfill
\begin{minipage}{.49\textwidth}
\centering
\includegraphics[width=0.8\linewidth]{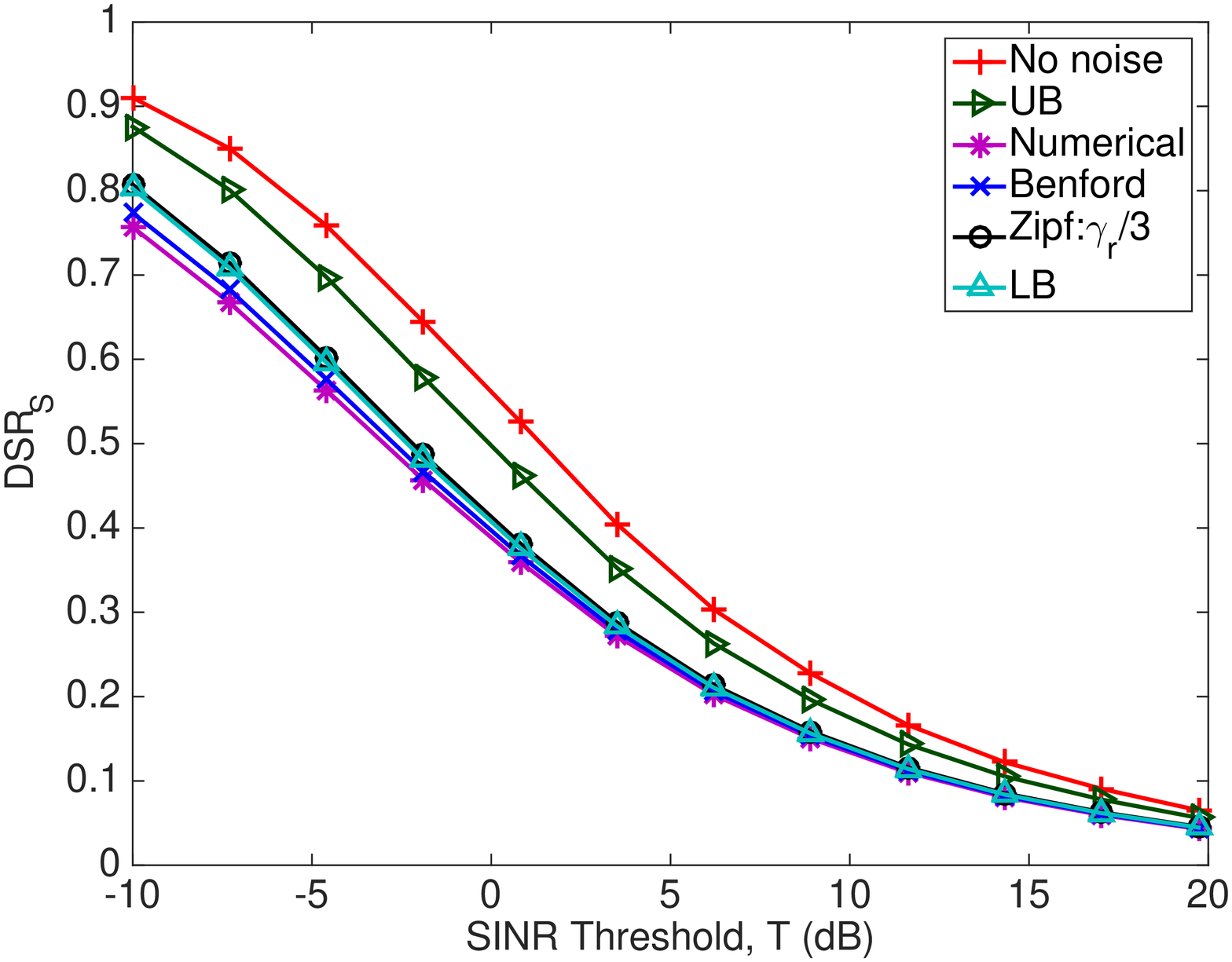}
\caption{\small{\!Bounds and approximations to the \!optimal \!$\DSRs$ for $M=10$, $\SNR=10, \lambda=1$, Zipf request pmf with $\gamma_r=0.5$.}}
\label{05SNR10new} 
\end{minipage}
\end{figure}

\begin{figure}[t!]
\centering
\begin{minipage}{.49\textwidth}
\centering
\includegraphics[width=0.8\linewidth]{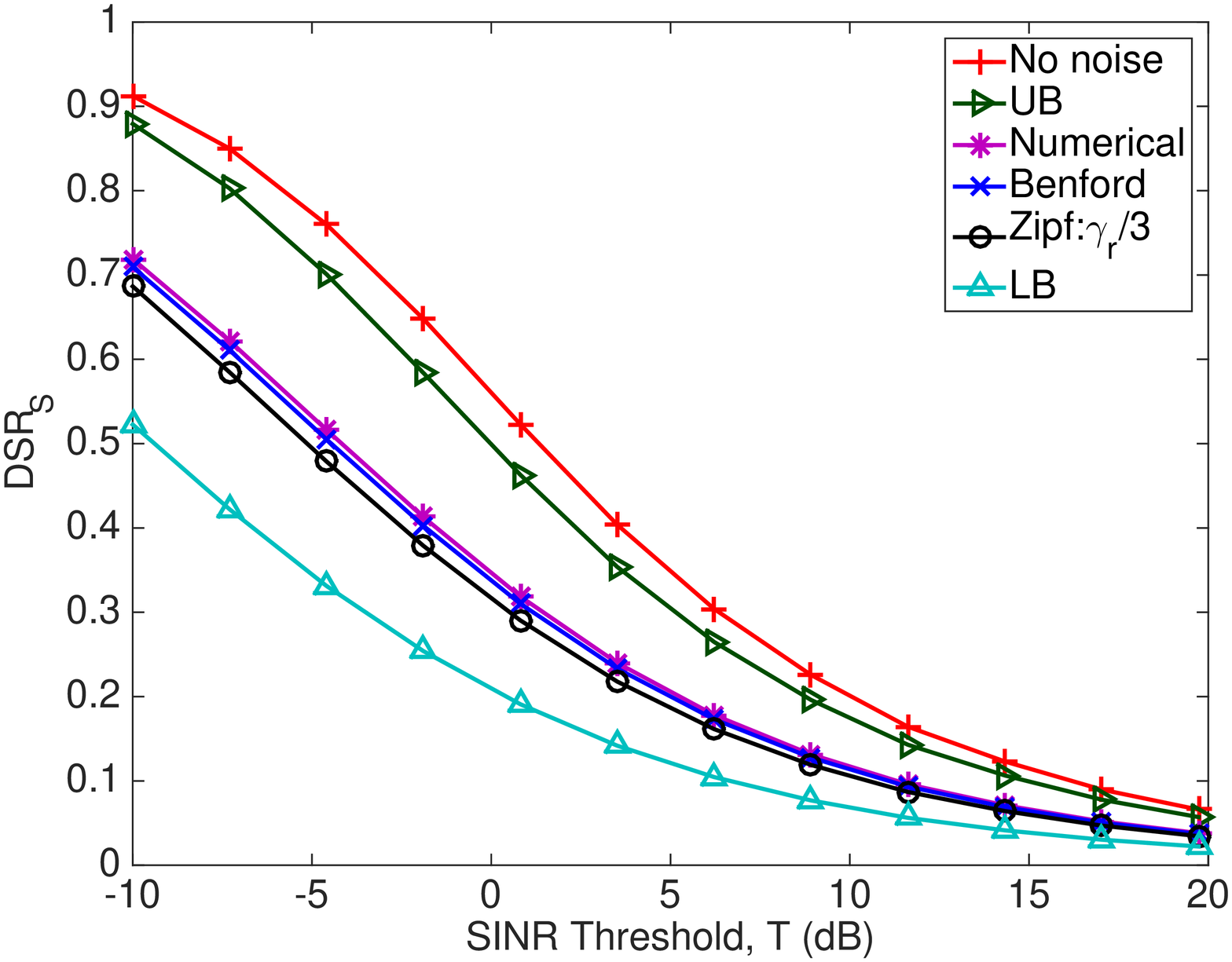}
\caption{\small{Bounds and approximations to the optimal $\DSRs$ for $M=10$, $\SNR=1, \lambda=1$, Zipf request pmf with $\gamma_r=2$.}}
\label{2SNR1new}
\end{minipage}
\hfill
\begin{minipage}{.49\textwidth}
\centering
\includegraphics[width=0.8\linewidth]{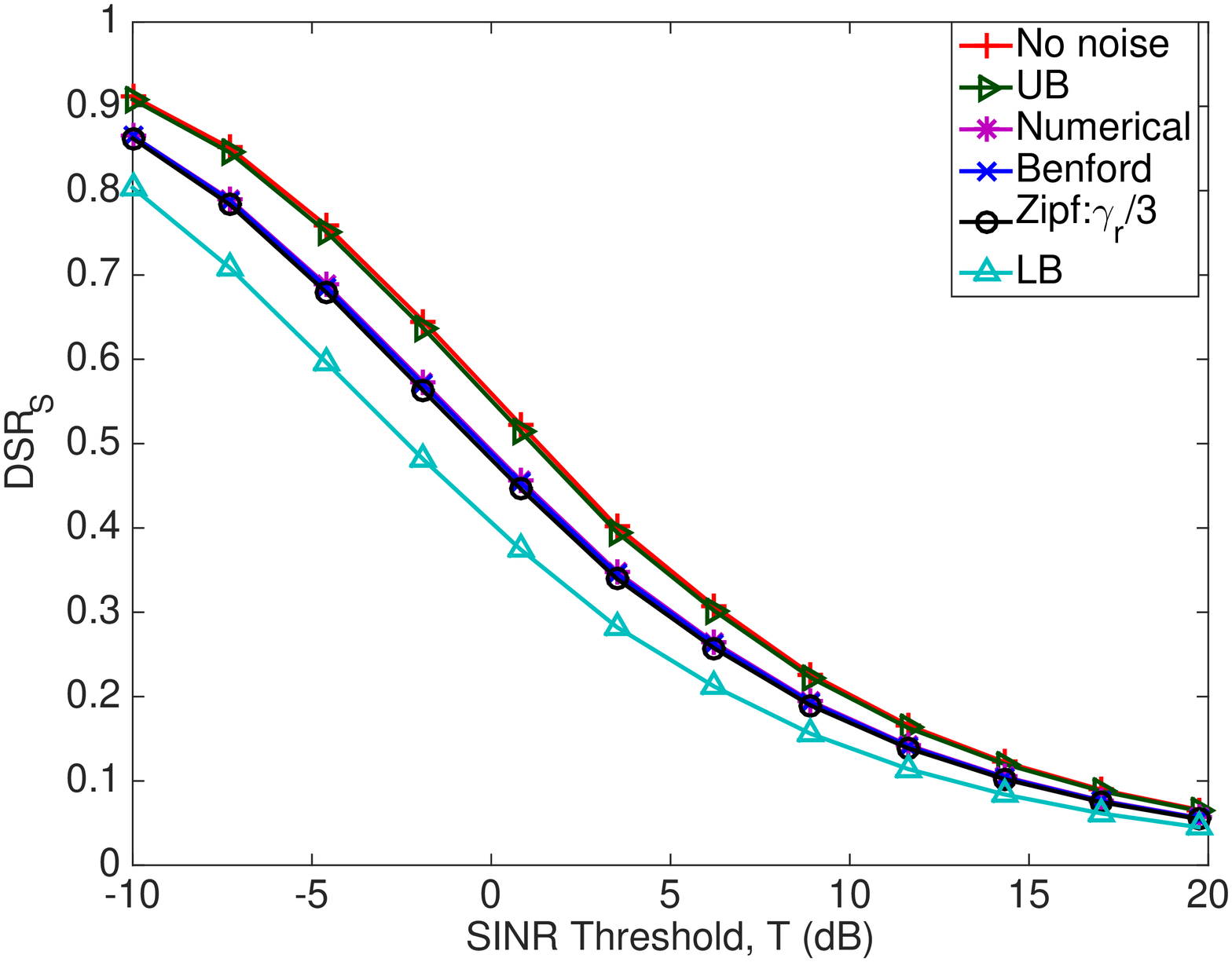}
\caption{\small{Bounds and approximations to the \!optimal $\DSRs$ for $M=10$, $\SNR=10, \lambda=1$, Zipf request pmf with $\gamma_r=2$.}}
\label{2SNR10new} 
\end{minipage}
\end{figure}

{\bf LB and UB get closer together as the $\SNR$ increases.}
In Fig. \ref{05SNR10new}, we compare the performance of the caching distributions for a Zipf request pmf with parameter $\gamma_r=0.5$ and $\SNR=10$. At high $\SNR$, the UB and LB are closer. Still, the numerical solution and the Zipf caching pmf with parameter $\gamma_r/3$ give similar densities of successful communication, which is very close to the lower bound because for that choice of $\gamma_r$, the request distribution converges to a uniform distribution. Benford caching distribution does not perform as well as the Zipf caching distribution, and is even worse than the LB. In Fig. \ref{2SNR1new}, where $\gamma_r=2$ and $\SNR=1$, the Zipf caching pmf with parameter $\gamma_r/3$ does not have the same performance as the optimal solution evaluated numerically. Benford distribution has also similar performance as the Zipf caching pmf. In Fig. \ref{2SNR10new}, we also show that Zipf caching pmf and Benford distributions have similar performance as the numerical solution for $\gamma_r=2$ and $\SNR=10$.

{\bf Transmit Diversity.} In the sequential serving model, where only one file is transmitted at a time network-wide, as discussed in Sect. \ref{MultipleFiles}, using a transmitter diversity scheme will improve the $\dsr$. For the second scenario presented in Sect. \ref{mostgeneralcase}, in which different files are transmitted simultaneously, a similar diversity scheme can be applied instead of treating the other transmitters as interferers. Diversity combining techniques include the maximal-ratio combining (e.g., of the $k$ closest transmitters \cite{Keeler2013}), where the received signals are weighted with respect to their $\SINR$ and then summed, equal-gain combining, where all the received signals are summed coherently, i.e., the shot-noise model \cite[Ch. 2]{BaccelliBook1}, and the selection combining, which is based on the strongest $\DD$ user association, in which the received signal power (e.g., from the $k$ strongest users \cite{Keeler2013}) is considered. 

Although diversity can decrease the outage probability, how to achieve this in practice is a critical issue. Diversity would seem to require synchronization of all transmitting devices at the physical layer unless higher layer coding is used, which might not be very practical for content distribution. Assuming full synchronization provides an upper bound on what could be achieved. Due to space constraints, we leave such analysis to future work.

\section{Conclusions}
\label{Conc}
Content distribution using direct $\DD$ communications is a promising approach for optimizing the utilization of air-interface resources in 5G network. This work is the first attempt to derive closed form expressions for the optimal content caching distribution and the optimal caching strategies providing maximum $\dsr$ in terms of the optimal fractions of transmitters and receivers in a $\DD$ network by using a homogeneous $\PPP$ model with realistic noise, interference and Rayleigh fading. We derive the $\SINR$ coverage for different transmission strategies in $\DD$ networks with some idealized modeling aspects, i.e., simultaneous scheduling of the users containing the same type of files and Zipf distributed content caching assumption for the general multi-file transmissions. Our results for the sequential transmission model show that the optimal caching pmf can also be modeled using the Zipf law and its exponent $\gamma_c$ is related to $\gamma_r$ through a simple expression involving the path loss exponent: $\gamma_c=\frac{\gamma_r}{(\alpha/2+1)}$. The optimal content placement for more general demand profiles under Rayleigh, Ricean and Nakagami fading distributions suggests to flatten the request distribution to optimize the caching performance.

The limitations of the model can be overcome by investigating the optimal caching distributions that maximize the $\dsr$ for the more general transmission settings incorporating the transmit diversity,  and developing intelligent scheduling techniques, which are left as future work. The dynamic settings capturing the changes in the file popularities over time and the interference caused by simultaneous transmissions should also be considered. Future issues include the minimization of backhaul transmissions and $\BS$ overhead to optimize resource utilization through $\DD$ collaboration. Future work could also include the design of distributed caching strategies to maximize the hit probability for users by using an $\SINR$ coverage model or a distance-based coverage process given the limited range of $\DD$.

\begin{appendix}
\subsection{Proof of Lemma \ref{txfrac}} 
\label{App:AppendixFirst}
The problem in (\ref{optimization1}) is equivalent to ${\DSR^*}= \underset{\gamma_1<a \leq 1}{\max} \quad\lambda\left(\frac{a-\gamma_1}{1-\gamma_1}\right) {\DSR_{\T}}(\gamma_1)$, where ${\DSR_{\T}}(\gamma_1)=\gamma_2 \pcov({\T},\lambda\gamma_1,\alpha)$ is the $\dsr$ for $a=1$, and ${\DSR_{\T}}(\gamma_1)$ can be rewritten as ${\DSR_{\T}}(\gamma_1)=(1-\gamma_1)\pi\lambda\gamma_1 \int_{0}^{\infty}{ e^{-\pi\lambda\gamma_1 r\beta({\T},\alpha)-\mu \T\sigma^2 r^{\frac{\alpha}{2}}} \, \mathrm{d}r}$. Taking the derivative with respect to $\gamma_1$, we obtain
\begin{eqnarray}
\label{diffeqn}
\frac{\partial {\DSR_{\T}}(\gamma_1)}{\partial \gamma_1}+\frac{(a-1)}{(1-\gamma_1)(a-\gamma_1)}{\DSR_{\T}}(\gamma_1)=0.
\end{eqnarray}
Using (\ref{pcov}), the solution of (\ref{diffeqn}) yields the following expression: $0<\frac{\left[\frac{1}{\gamma_1}-\frac{1}{a-\gamma_1}\right]}{\pi \lambda \beta({\T},\alpha)}=\frac{\int_{0}^{\infty}{ e^{-\pi\lambda\gamma_1 r\beta({\T},\alpha)-\mu \T\sigma^2 r^{\frac{\alpha}{2}}}r \, \mathrm{d}r}}{\int_{0}^{\infty}{ e^{-\pi\lambda\gamma_1 r\beta({\T},\alpha)-\mu \T\sigma^2 r^{\frac{\alpha}{2}}} \, \mathrm{d}r}}$. Thus, the positivity condition yields $a-\gamma_1>\gamma_1$, implying $\gamma_1<a/2$ and $\gamma_2>a/2$.


\subsection{Proof of Lemma \ref{arbitrarynoiselemma}}
\label{App:Appendix0}
We use the relation $\int_{0}^{\infty}{\, e^{-(x-b)^2}x\mathrm{d}x}
=\frac{1}{2}e^{-b^2}+\sqrt{\pi}b{\Q}(-\sqrt{2}b)$, which follows from a change of variables $u=x-b$ and separating into two integrals, where ${\Q}(x)=\frac{1}{\sqrt{2\pi}}\int_{x}^{\infty}{\, e^{-y^2/2}\mathrm{d}y}$ is the standard Gaussian tail probability. Using this relation, $\pcov({\T},\lambda\gamma_1,4)$ is given as
\begin{eqnarray}
\label{complicatedintegral2}
\pcov({\T},\lambda\gamma_1,4)=\pi\lambda\gamma_1\int_{0}^{\infty}{ e^{-\pi\lambda\gamma_1 r\beta({\T},4)-\mu \T\sigma^2 r^{2}} \, \mathrm{d}r}
=\frac{\pi^{\frac{3}{2}}\lambda\gamma_1}{\sqrt{\mu \T\sigma^2}}e^{\frac{(\lambda\gamma_1\pi\beta({\T},4))^2}{4\mu \T\sigma^2}}{\Q}\Big(\frac{\lambda\gamma_1\pi\beta({\T},4)}{\sqrt{2\mu \T\sigma^2}}\Big).
\end{eqnarray}

The solution of (\ref{diffeqn}) in Appendix \ref{App:AppendixFirst} for $\alpha=4$ yields the following expression:
\begin{eqnarray}
\label{complicatedintegral3}
\frac{\left[\frac{1}{\gamma_1}-\frac{1}{a-\gamma_1}\right]}{\pi \lambda \beta({\T},4)}
=\frac{e^{\mu T\sigma^2 u_0^2}\int_{0}^{\infty}{ e^{-\mu \T\sigma^2 \left(r+u_0\right)^2}r \, \mathrm{d}r}}{e^{\mu T\sigma^2 u_0^2}\int_{0}^{\infty}{ e^{-\mu \T\sigma^2 \left(r+u_0\right)^2} \, \mathrm{d}r}}
\,\stackrel{(a)}{=}\,\frac{\int_{u_0}^{\infty}{ e^{-\mu \T\sigma^2 u^2}u \, \mathrm{d}u}-u_0\int_{u_0}^{\infty}{ e^{-\mu \T\sigma^2 u^2} \, \mathrm{d}u}}{\int_{u_0}^{\infty}{ e^{-\mu \T\sigma^2 u^2} \, \mathrm{d}u}}\nonumber\\
\stackrel{(b)}{=}\frac{\frac{1}{2\mu\T\sigma^2}\int_{\mu \T\sigma^2u_0^2}^{\infty}{ e^{-v_1} \, \mathrm{d}v_1}-\frac{\pi\lambda\gamma_1\beta({\T},4)}{2(\mu \T\sigma^2)^{1.5}}\int_{\sqrt{\mu \T\sigma^2}u_0}^{\infty}{ e^{-v_2^2} \, \mathrm{d}v_2}}{\frac{1}{\sqrt{\mu \T\sigma^2}}\int_{\sqrt{\mu \T\sigma^2}u_0}^{\infty}{ e^{-v_2^2} \, \mathrm{d}v_2}}
\stackrel{(c)}{=}\,\frac{\pi\lambda\gamma_1}{2\mu\T\sigma^2}\left[\pcov({\T},\lambda\gamma_1,4)^{-1}-\beta({\T},4)\right],
\end{eqnarray}
where $u_0=\frac{\pi\lambda\gamma_1\beta({\T},4)}{2\mu \T\sigma^2}$, $(a)$ follows from employing a change of variables $u=r+\frac{\pi\lambda\gamma_1\beta({\T},4)}{2\mu \T\sigma^2}$, $(b)$ also follows from employing change of variables $v_1=\mu\T\sigma^2u^2$ and $v_2=\sqrt{\mu\T\sigma^2}u$, and $(c)$ follows from employing the definition of standard Gaussian tail probability, and employing the definition of the coverage probability for Rayleigh fading, noise and $\alpha=4$ \cite{Andrews2011}. 

The relation in (\ref{complicatedintegral3}) can be rearranged to obtain the following equality:
\begin{eqnarray}
\label{gammaeq}
\frac{1}{\gamma_1}\left[\frac{1}{\gamma_1}-\frac{1}{a-\gamma_1}\right]=\frac{(\pi \lambda)^2 \beta({\T},4)^2}{{2\mu \T\sigma^2}}\left[\frac{\beta({\T},4)^{-1}}{\pcov({\T},\lambda\gamma_1,4)}-1\right]\substack{(a) \\ >}\, 0,
\end{eqnarray}
where ($a$) follows from the fact that $\pcov({\T},\lambda\gamma_1,4)\leq \beta({\T},4)^{-1}$, where $\beta({\T},4)^{-1}$ is the no noise coverage. Given that $\gamma_1$ is optimal, i.e., it satisfies (\ref{gammaeq}), the maximum of the $\dsr$ $\pcov({\T},\lambda\gamma_1,4)$ is 
\begin{eqnarray}
\label{optimalpcovalpha4}
\pcov({\T},\lambda\gamma_1,4)=\left(\frac{1}{\gamma_1}\Big[\frac{1}{\gamma_1}-\frac{1}{a-\gamma_1}\Big]\frac{{2\mu \T\sigma^2}}{(\pi \lambda)^2 \beta({\T},4)}+\beta({\T},4)\right)^{-1},
\end{eqnarray}
using which ${\DSR^*}$ can be obtained. Combining (\ref{optimalpcovalpha4}) with (\ref{complicatedintegral2}), the optimal value of $\gamma_1$ is found.

\subsection{Proof of Lemma \ref{maximumcoveragerayleighlemma}}\label{App:Appendix1}
The problem is equivalent to the formulation ${\DSR^*}=\underset{\gamma_1}{\max}\,\lambda(1-\gamma_1)\big[\frac{1}{\beta({\T},\alpha)}-\frac{\mu {\T}\sigma^2\left(\lambda\gamma_1\pi\right)^{-\frac{\alpha}{2}}}{{\beta({\T},\alpha)}^{\frac{\alpha}{2}+1}}\Gamma\left(1+\frac{\alpha}{2}\right)+o\left(\sigma^2\right)\big]$, which is concave in $\gamma_1$. Hence, taking its derivative with respect to $\gamma_1$, we obtain $\frac{1}{\beta({\T},\alpha)}+o\left(\sigma^2\right)=\frac{\mu \T\sigma^2(\lambda\pi)^{-\frac{\alpha}{2}}}{{\beta({\T},\alpha)}^{\frac{2}{\alpha}+1}}\Gamma\left(1+\frac{\alpha}{2}\right)\left[\frac{\alpha}{2}\gamma_1^{-1}+\left(1-\frac{\alpha}{2}\right)\right]\gamma_1^{-\frac{\alpha}{2}}$. The solution $\gamma_1^*$ satisfies the polynomial equation $\big[\frac{2{\beta({\T},\alpha)}^{\frac{2}{\alpha}}}{2-\alpha}\frac{\left(\lambda \pi\right)^{\frac{\alpha}{2}}}{\mu \T\sigma^2\Gamma\left(1+\frac{\alpha}{2}\right)}\big]\left(\gamma_1^*\right)^{\frac{\alpha}{2}+1}o(\sigma^2)-\gamma_1^*= \frac{\alpha}{2-\alpha}$, using which ${\DSR^*}$ is found.

\subsection{Proof of Lemma \ref{maxcoversmallnoise}}
\label{App:AppendixA}
Using (\ref{diffeqn}), and letting $\alpha=4$, we have
\begin{eqnarray}
\label{alpha4approx}
\frac{\left[\frac{1}{\gamma_1}-\frac{1}{a-\gamma_1}\right]}{\pi \lambda \beta({\T},4)}&=&\frac{\int_{0}^{\infty}{ e^{-\pi\lambda\gamma_1 r\beta({\T},4)-\mu \T\sigma^2 r^2}r \, \mathrm{d}r}}{\int_{0}^{\infty}{ e^{-\pi\lambda\gamma_1 r\beta({\T},4)-\mu \T\sigma^2 r^2} \, \mathrm{d}r}}
\stackrel{(a)}{=} \frac{\int_{0}^{\infty}{ e^{-\pi\lambda\gamma_1 r\beta({\T},4)}(1-\mu \T\sigma^2 r^2+o(\sigma^2))r \, \mathrm{d}r}}{\int_{0}^{\infty}{ e^{-\pi\lambda\gamma_1 r\beta({\T},4)}(1-\mu \T\sigma^2 r^2+o(\sigma^2)) \, \mathrm{d}r}},\nonumber\\
&=&\frac{(\pi\lambda\gamma_1 \beta({\T},4))^2-6(\mu \T\sigma^2)+o(\sigma^2)}{{(\pi\lambda\gamma_1 \beta({\T},4))}^3-2(\mu \T\sigma^2)(\pi\lambda\gamma_1 \beta({\T},4))+o(\sigma^2)},
\end{eqnarray}
where $(a)$ follows from $\exp(-x)=1-x+o(x)$ for $x\rightarrow 0$. Then, using (\ref{alpha4approx}), we obtain 
\begin{eqnarray}
\label{gamma1relation}
\frac{2a}{\gamma_1}=\frac{(\pi\lambda\beta({\T},4))^3\gamma_1^2+(\mu \T\sigma^2)(\pi\lambda\beta({\T},4))2+o(\sigma^2)}{(\mu \T\sigma^2)(\pi\lambda\beta({\T},4))2+o(\sigma^2)}
=\frac{(\pi\lambda\gamma_1\beta({\T},4))^2+2\mu \T\sigma^2+o(\sigma^2)}{2\mu \T\sigma^2+o(\sigma^2)}.
\end{eqnarray}
	
Given that $\gamma_1$ is optimal, i.e., it satisfies (\ref{gamma1relation}), $\pcov({\T},\lambda\gamma_1,4)$ is given by
\begin{eqnarray}
\pcov({\T},\lambda\gamma_1,4)=\frac{1}{\beta({\T},4)}-\frac{2\mu \T\sigma^2\left(\lambda\gamma_1\pi\right)^{-2}}{\beta({\T},4)^{3}}+o\left(\sigma^2\right)
=\frac{2}{\beta({\T},4)}\left[\frac{\mu \T\sigma^2(a-\gamma_1)+o(\sigma^2)}{\mu \T\sigma^2(2a-\gamma_1)+o(\sigma^2)}\right]+o(\sigma^2),\nonumber
\end{eqnarray}
using which the final result can be obtained. As $\sigma^2\to 0$, $\underset{\sigma^2\to 0}{\lim}\quad \pcov({\T},\lambda\gamma_1,4)=\frac{2}{\beta({\T},4)}\left[\frac{a-\gamma_1}{2a-\gamma_1}\right]$.

\subsection{Proof of Lemma \ref{smallnoiselemma}}
\label{App:AppendixB}
For small but non-zero noise case, and given that $\sum\nolimits_{i=1}^{M} p_c(i)=1$, (\ref{main-opt}) can be rewritten as
\begin{equation}
\begin{aligned}
& \underset{p_c}{\max}
& & \lambda\Big[\frac{1}{\beta({\T},\alpha)}+o\left(\sigma^2\right)\Big]-\Gamma\left(1+\frac{\alpha}{2}\right)\Big[\frac{\mu \T\sigma^2\left(\lambda\pi\right)^{-\alpha/2}}{\beta({\T},\alpha)^{\alpha/2+1}}\Big]\sum\limits_{i=1}^M {\lambda p_r(i)p_c(i)^{-\alpha/2}},\nonumber 
\end{aligned}
\end{equation}
equivalent to minimizing $\sum\nolimits_{i=1}^M {p_r(i)p_c(i)^{-\frac{\alpha}{2}}}=\big(\sum\nolimits_{j=1}^{M}{\frac{1}{j^{\gamma_r}}}\big)^{-1}\sum\nolimits_{i=1}^M {\frac{p_c(i)^{-\frac{\alpha}{2}}}{i^{\gamma_r}}}$ subject to $\sum\nolimits_{i=1}^{M} p_c(i)=1$.

Using the Lagrange multiplier method \cite{Bertsekas1999}, we define, $\Lambda({\bf p_c},\eta)=\sum\limits_{i=1}^{M}{\frac{1}{i^{\gamma_r}}p_c(i)^{-\alpha/2}}+\eta\Big(\sum\limits_{i=1}^M{p_c(i)}-1\Big)$, where ${\bf p_c}=[p_c(1) \hdots p_c(M)]$, and $\eta$ is the Lagrange multiplier. To maximize $\Lambda({\bf p_c},\eta)$, we take its partial derivatives with respect to ${\bf p_c}$. The partial derivative $\frac{\partial\Lambda({\bf p_c},\eta)}{\partial\eta}$ reduces to the constraint equation. 

Partial derivative of $\Lambda({\bf p_c},\eta)$ with respect to $p_c(i)$ gives $(-\alpha/2)\frac{1}{i^{\gamma_r}}p_c(i)^{-\alpha/2-1}+\eta=0,\quad i=1,\hdots, M$. Hence, for any $i\neq j$ pair of file indices, we require $\big(-\frac{\alpha}{2}\big)\frac{1}{i^{\gamma_r}}p_c(i)^{-\frac{\alpha}{2}-1}=\big(-\frac{\alpha}{2}\big)\frac{1}{j^{\gamma_r}}p_c(j)^{-\frac{\alpha}{2}-1}$, implying that ${p_c(j)}/{p_c(i)}=(i/j)^{\frac{\gamma_r}{\alpha/2+1}}$. Then, $p_c(i)$ also has Zipf distribution with exponent $\gamma_c=\frac{\gamma_r}{\alpha/2+1}$.

\subsection{Proof of Lemma \ref{ArbitraryNoise}}
\label{App:AppendixC}
We investigate the general solution of (\ref{main-opt}). Using the Lagrange multiplier method \cite{Bertsekas1999}, we define $\Lambda({\bf p_c},\eta)=\sum\nolimits_{i=1}^M {\lambda_t p_r(i) \pcov({\T},\lambda_t p_c(i),\alpha)}+\eta\Big(\sum\nolimits_{i=1}^{M} p_c(i)-1\Big)$. The partial derivatives of $\Lambda({\bf p_c},\eta)$ with respect to $p_c(i)$ for $i=1,\hdots, M$ give M equations.
\begin{eqnarray}
\label{partial_pc}
\frac{\partial \Lambda({\bf p_c},\eta)}{\partial p_c(i)}&=&\lambda_t p_r(i)\frac{\partial \pcov({\T},\lambda_t p_c(i),\alpha)}{\partial p_c(i)}+\eta
=\lambda_t p_r(i)\Big[\pi\lambda_t \int\nolimits_{0}^{\infty}{ e^{-\pi\lambda_t p_c(i) r\beta({\T},\alpha)-\mu {\T}\sigma^2 r^{\alpha/2}} \, \mathrm{d}r}\nonumber\\
&-&(\pi\lambda_t)^2\beta({\T},\alpha) p_c(i)\int\nolimits_{0}^{\infty}{ e^{-\pi\lambda_t p_c(i) r\beta({\T},\alpha)-\mu {\T}\sigma^2 r^{\alpha/2}}r \, \mathrm{d}r}\Big]+\eta.
\end{eqnarray}
To maximize $\Lambda({\bf p_c},\eta)$, we equate the $\RHS$ of (\ref{partial_pc}) to 0 and obtain
\begin{eqnarray}
\label{der_pc}
 \int\nolimits_{0}^{\infty}{[1-\pi\lambda_t \beta({\T},\alpha) p_c(i)r] e^{-\pi\lambda_t p_c(i) r\beta({\T},\alpha)-\mu {\T}\sigma^2 r^{\alpha/2}} \, \mathrm{d}r}=-\frac{\eta}{p_r(i) \pi\lambda_t^2}.
\end{eqnarray}
The partial derivative of $\pcov({\T},\lambda_t p_c(i),\alpha)$ with respect to $\lambda_t$ is given as
\begin{eqnarray}
\label{lambda-pc}
\frac{\partial {\pcov}({\T},\lambda_t p_c(i),\alpha)}{\partial \lambda_t}=\pi p_c(i) \int\nolimits_{0}^{\infty}{ e^{-\pi\lambda_t p_c(i) r\beta({\T},\alpha)-\mu {\T}\sigma^2 r^{\alpha/2}} \, \mathrm{d}r}\\
-(\pi p_c(i))^2 \beta({\T},\alpha) \lambda_t\int\nolimits_{0}^{\infty}{ e^{-\pi\lambda_t p_c(i) r\beta({\T},\alpha)-\mu {\T}\sigma^2 r^{\alpha/2}}r \, \mathrm{d}r}
=\frac{p_c(i)}{\lambda_t}\frac{\partial \pcov({\T},\lambda_t p_c(i),\alpha)}{\partial p_c(i)}.\nonumber
\end{eqnarray}
Combining the relations (\ref{der_pc}) and (\ref{lambda-pc}) results in $\frac{\partial \pcov({\T},\lambda_t p_c(i),\alpha)}{\partial \lambda_t}=-\eta\frac{p_c(i)}{\lambda_t^2 p_r(i)}$. Using the definition of $\pcov({\T},\lambda_t p_c(i),\alpha)$, we can easily note that $\pcov({\T},\lambda_t p_c(i),\alpha)=\pcov({\T}\left({p_c(j)}/{p_c(i)}\right)^{\alpha/2},\lambda_t p_c(j),\alpha)$. Taking the derivative of this expression with respect to $\lambda_t$, we have
\begin{eqnarray}
\label{pcov-i-j}
\frac{\partial \pcov({\T},\lambda_t p_c(j),\alpha)}{\partial \lambda_t}=-\eta\frac{p_c(j)}{\lambda_t^2 p_r(j)}
=\frac{\partial \pcov({\T},\lambda_t p_c(i),\alpha)}{\partial \lambda_t}\frac{p_r(i)/p_c(i)}{p_r(j)/p_c(j)}.
\end{eqnarray}
We can rewrite (\ref{pcov-i-j}) using the expression for $\pcov({\T},\lambda_t p_c(i),\alpha)$ as follows
\begin{align}
\label{derivative-relation}
\frac{\partial \pcov({\T},\lambda_t p_c(j),\alpha)}{\partial \lambda_t}
=\frac{\partial \pcov({\T}\left(\frac{p_c(j)}{p_c(i)}\right)^{\alpha/2},\lambda_t p_c(j),\alpha)}{\partial \lambda_t}\frac{p_r(i)/p_c(i)}{p_r(j)/p_c(j)}.
\end{align}

Next, by employing a change of variables $v=r\beta({\T},\alpha)$, we can rewrite (\ref{pcov}) in Definition \ref{pcov} as
\begin{eqnarray}
\label{covinvariancebeta}
\pcov({\T},\lambda_t,\alpha)=\frac{\pi\lambda_t}{\beta({\T},\alpha)}\int\nolimits_{0}^{\infty}{ e^{-\pi\lambda_t v-\mu\left[\frac{\T}{\beta({\T},\alpha)^{\alpha/2}}\right]\sigma^2 v^{\alpha/2}} \, \mathrm{d}v}.
\end{eqnarray}
We investigate the relation between $\beta({\T},\alpha)^{\alpha/2}$ and $\T$ in Fig. \ref{betavsT}, for practical $\alpha$ and $\mu$ values, and observe the linear dependence, where the slope is mainly determined by $\alpha$, and changes only slightly by varying $\mu$. Based on these simulations, since $\beta({\T},\alpha)^{\alpha/2}/{\T}$ is invariant to $\T$ and using the relation in (\ref{covinvariancebeta}), it is reasonable to write $\pcov({\T},\lambda_t p_c(j),\alpha)$ as a separable function which is the form $f(\lambda_t p_c(j),\alpha)g({\T})$. By taking its derivative with respect to $\lambda_t$, we can then rewrite (\ref{derivative-relation}) as $g({\T})=g\Big({\T}\Big(\frac{p_c(j)}{p_c(i)}\Big)^{\alpha/2}\Big)\frac{p_c(j)}{p_c(i)}\Big(\frac{j}{i}\Big)^{\gamma_r}$. Taking the derivative of both sides with respect to $\T$, we obtain $\frac{d g({\T})}{d \T}=\Big(\frac{p_c(j)}{p_c(i)}\Big)^{\alpha/2}\frac{d g({\T})}{d \T}\frac{p_c(j)}{p_c(i)}\Big(\frac{j}{i}\Big)^{\gamma_r}$, implying that $p_c(j)/p_c(i)=\left(i/j\right)^{\frac{\gamma_r}{\alpha/2+1}}$. Then, $p_c(\cdot)$ is also Zipf($\gamma_c$) distributed with $\gamma_c=\frac{\gamma_r}{\alpha/2+1}$.

\subsection{Proof of Lemma \ref{ArbitraryNoiseApprox}}
\label{App:AppendixD}
The coverage probability $\pcov({\T},\lambda_t p_c(i),4)$ for Rayleigh fading and general noise with $\alpha=4$ is $\pcov({\T},\lambda_t p_c(i),4)=\frac{\pi^{1/2}\sqrt{2}}{\beta({\T},4)}xQ(x)\exp(x^2/2)$, where $x=\frac{\pi \lambda_t p_c(i) \beta({\T},4)}{\sqrt{2\mu \T \sigma^2}}$. The details of the derivation follow from \cite{Andrews2011}. We approximate $\pcov({\T},\lambda_t p_c(i),4)$ by using the following tight approximation for $\Q$-function in \cite{Karagiannidis2007} as ${\Q}(x)\approx \frac{(1-\exp(-1.4x))\exp(-x^2/2)}{1.135\sqrt{2\pi}x},\,\, x>0$. Hence, $\pcov({\T},\lambda_t p_c(i),4)\approx\frac{1}{1.135\beta({\T},4)}\big(1-\exp\big(-\frac{\pi \lambda_t p_c(i) \beta({\T},4)}{\sqrt{\mu \T \sigma^2}}\big)\big)$. Using the Lagrange multiplier method \cite{Bertsekas1999} to find the solution of the maximum $\dsr$ problem defined in (\ref{main-opt}), for the file indices $i$ and $j$, we obtain the relation $\frac{1}{i^{\gamma_r}}\exp{(-\frac{\pi \lambda_t p_c(i) \beta({\T},4)}{\sqrt{\mu \T \sigma^2}})}=\frac{1}{j^{\gamma_r}}\exp{(-\frac{\pi \lambda_t p_c(j) \beta({\T},4)}{\sqrt{\mu \T \sigma^2}})}$ that yields the following difference between the file caching probabilities as a function of the network parameters, Zipf exponent $\gamma_r$ and the file indices, which is given as $p_c(i)-p_c(j)=-\frac{\sqrt{\mu \T \sigma^2}\gamma_r}{\pi \lambda_t\beta({\T},4)}\log\left(\frac{i}{j}\right)$. Using the relation $\sum\nolimits_{i=1}^{m}{p_c(i)}=1$, 
the caching distribution is found as $p_c(i)=\frac{1}{M}+\frac{\sqrt{\mu \T \sigma^2}\gamma_r}{M\pi \lambda_t\beta({\T},4)}\sum\nolimits_{j=1}^M{\log\left(\frac{j}{i}\right)}$, $i=1, \hdots, M$, and the final result can be obtained by rearranging the terms.
The required condition for the pmf $p_c(\cdot)$ to be valid is $p_c(i)>0$, for $i=1, \hdots, M.$ Since $p_c(i)$ is decreasing in $i$, a sufficient condition is $p_c(M)=\frac{1}{M}+\frac{\sqrt{\mu \T \sigma^2}\gamma_r}{M\pi \lambda_t\beta({\T},4)}\sum\limits_{j=1}^M{\log\left(\frac{j}{M}\right)}\geq 0$. Thus, for a total number of files $M$, we require that ${\sqrt{\mu \T \sigma^2}\gamma_r}\leq {\pi \lambda_t\beta({\T},4)}[M\log(M)-\log(M!)]^{-1}$.

\subsection{Proof of Lemma \ref{mul-file-optlemma}}
\label{AppendixMINdsr}
The optimization formulation to maximize the $\dsr$ of the least popular file is given as 
\begin{equation}
\label{optimization-maxmin-initial}
\begin{aligned}
\underset{\rho_i}{\max}\hspace{0.3cm}&\underset{i}{\min}
& & {\xi}^{-1}p_r(i)\rho_ip_c(i)\pcov({\T},\lambda_t\xi,\alpha),
\end{aligned}
\end{equation}
where $\xi=\sum\nolimits_{i=1}^{M}{\rho_i p_c(i)}$, $0\leq \rho_i \leq 1$. We can calculate $\pcov({\T},\lambda_t\xi,\alpha)$ if $p_r(i)$, $p_c(i)$ and $\xi$ are known priorily. The caching distribution can be modeled by $p_c(i)\sim \mathrm{Zipf}(\gamma_c)$, where $\gamma_c=\gamma_r/(\alpha/2+1)$ based on Lemma \ref{ArbitraryNoise}. Letting $\xi_i=\rho_ip_c(i)$, the formulation in (\ref{optimization-maxmin-initial}) is equivalent to the following: 
\begin{equation}
\label{wcaching}
\begin{aligned}
&\underset{0\leq \xi_i\leq p_c(i)}{\max}
& & w \\
& \quad\text{s.t.}
& & w\leq p_r(i)\xi_i,\quad i=1,\hdots,M.
\end{aligned}
\end{equation}
The optimal solution of (\ref{wcaching}) is $\xi_i=p_c(i)$ if $\xi=1$ since $\xi_i$ cannot be larger than $p_c(i)$. Defining $\xi=\sum\nolimits_{i=1}^{M}{\xi_i}$, if $\xi<1$, then, $\xi_i\leq p_c(i)$. The optimal solution can be found by equating $p_r(i)\xi_i$ for all $i \in\{1,\hdots,M\}$ so that the least desired file with $p_r(m)$ is multiplied by the highest $\xi_m$. Then, the required condition is $p_r(i)\xi_i=p_r(j)\xi_j$ for $i\neq j$, which is equivalent to $\frac{\rho_i}{\rho_j}=\frac{p_r(j)p_c(j)}{p_r(i)p_c(i)}$ for $i\neq j$. Using the constraint of (\ref{wcaching}), we obtain $\xi=\sum\nolimits_{j=1}^{M}{\rho_jp_c(j)}=\rho_ip_c(i)+\rho_ip_c(i)\sum\nolimits_{j=1,j\neq i}^{M}{\frac{p_r(i)}{p_r(j)}}$, and solving for $\rho_i$, for all $i\in \{1,\hdots,M\}$, we get $\rho_i=\frac{\xi}{p_c(i)}\Big(1+\sum\nolimits_{j=1,j\neq i}^{M}{\frac{p_r(i)}{p_r(j)}}\Big)^{-1}={\xi}\Big/{\sum\nolimits_{j=1}^{M}{\frac{p_c(i)p_r(i)}{p_r(j)}}}$. Hence, the optimal $\rho_i$'s for $i\in\{1,\hdots,M\}$ should satisfy $p_r(i)p_c(i)\rho_i=\eta$, if $\rho_i<1$, and $p_r(i)p_c(i)<\eta$, if $\rho_i=1$, implying that for any $\rho_i<1$, $\eta<p_r(i)p_c(i)$, yielding $\eta< \underset{i,\,\, \rho_i<1}{\min}\, {p_r(i)p_c(i)}$, and for any $\rho_i=1$, $\eta> \underset{i,\,\, \rho_i=1}{\max}\, {p_r(i)p_c(i)}$ for some constant $\eta$. Hence, the objective of (\ref{optimization-maxmin-initial}) can be rewritten as
\begin{eqnarray}
\label{optimization-maxmin-initial2}
{\xi}^{-1}p_r(i)\rho_ip_c(i)\pcov({\T},\lambda_t\xi,\alpha)=\frac{\pcov({\T},\lambda_t\xi,\alpha)}{\sum\nolimits_{j=1}^{M}{\frac{1}{p_r(j)}}}\begin{cases}
=\frac{\eta}{\xi}\pcov\Big({\T},\lambda_t\xi,\alpha\Big), \, \, \rho_i<1\\
<\frac{\eta}{\xi}\pcov\Big({\T},\lambda_t\xi,\alpha\Big), \, \, \rho_i=1
\end{cases},\, i=1,\hdots,M,
\end{eqnarray}
which is increasing in $\eta$. The solution of (\ref{optimization-maxmin-initial}) can be found by letting $\eta=\underset{i,\,\,\rho_i=1}{\max}\,\, {p_r(i)p_c(i)}$, implying that $\rho_i=1$ for $i\geq j$ for some $j$, and for $1\leq i\leq j-1$, we have $\rho_i=\frac{\eta}{p_r(i)p_c(i)}=\frac{p_r(j)p_c(j)}{p_r(i)p_c(i)}<1$.

\subsection{Proof of Lemma \ref{multifileopt2lemma}}
\label{AppendixMAXdsr}
The formulation to maximize the $\dsr$ for all files is given by
\begin{equation}
\label{optimization-sum-initial}
\begin{aligned}
\DSR_{\tot}^{*}=& \underset{\rho_i}{\max}
& & {\xi}^{-1}\pcov({\T},\lambda_t\xi,\alpha)\sum\nolimits_{i=1}^{M}{p_r(i)\rho_ip_c(i)},
\end{aligned}
\end{equation}
where $\xi=\sum\limits_{i=1}^{M}{\rho_i p_c(i)}$, $0\leq \rho_i \leq 1$. For $\xi=1$, as $\xi_i\leq p_c(i)$ and $\sum_{i=1}^{M}{p_r(i)\xi_i}\leq \sum_{i=1}^{M}{p_r(i)p_c(i)}$, the optimal solution is $\xi_i=p_c(i)$, and $\rho_i=1$ for all $i$. For $\xi<1$, the solution of the maximum $\dsr$ problem can be found using water-filling, where $1\geq\rho_1\geq \rho_2\geq \hdots \geq \rho_M$ as $p_r(i)p_c(i)>p_r(j)p_c(j)$ for $i<j$. Using the constraint equation $\xi=\sum\nolimits_{i=1}^{M}{\rho_i p_c(i)}$, we have $\xi-\rho_1 p_c(1)=\sum\nolimits_{i=2}^M{\rho_i p_c(i)}$. To simplify the notation, we let $\widetilde{\xi}_i=\sum\nolimits_{j=i}^M{\rho_j p_c(j)}$. We propose the following update mechanism to determine the optimal $\rho_i$'s to find the optimal solution of (\ref{optimization-sum-initial}), where $\rho_i^*$ and $\widetilde{\xi_i}$ are updated as
\begin{eqnarray}
\rho_i^*=\underset{\rho_i\leq \rho_{i-1}^*}{\arg\min}\quad (\widetilde{\xi_i}-\rho_ip_r(i))^{+}, \quad
\widetilde{\xi}_i=\widetilde{\xi}_{i-1}-\rho_{i-1}^*p_c(i-1),\quad i>1,\nonumber
\end{eqnarray}
where $y^+=\max \{y,0\}$ and $\rho_{i-1}^*=1$. Since ${\xi}^{-1}\pcov({\T},\lambda_t\xi,\alpha)$ is decreasing in $\xi$, and $\sum\nolimits_{i=1}^{M}{p_r(i)\rho_ip_c(i)}$ is an increasing function of $\xi$, (\ref{optimization-sum-initial}) has an optimal solution. Now, consider the following formulation:
\begin{equation}
\label{box}
\begin{aligned}
& \underset{\rho_i}{\max}
& & \sum\nolimits_{i=1}^{M}{p_r(i)\rho_ip_c(i)}
\end{aligned}
\end{equation}
subject to $\xi=\sum\nolimits_{i=1}^{M}{\rho_i p_c(i)}$, where $0\leq \rho_i \leq 1$. Through a duality argument \cite{Boyd2009}, it is trivial to show that the solution of (\ref{box}) is $\rho_i=\xi$, $i\in\{1,\hdots,M\}$. Using which the formulation (\ref{optimization-sum-initial}) is upper bounded by ${\DSR_{\tot}^{*}} \leq \max\nolimits_{\xi} \, {\pcov({\T},\lambda_t\xi,\alpha)} \sum\nolimits_{i=1}^{M} p_r(i) p_c(i)$, and the upper bound is achieved for $\xi=1$.

\subsection{Proof of Theorem \ref{mainProbCov}} 
\label{App:AppendixA0}
The probability of coverage of a typical randomly located user $j$ is given by
\begin{eqnarray}
\label{probcoverage}
{\PCOV}({\T},\lambda_j,\alpha)=\mathbb{P}({\SINR}_j>{\T})=\int\nolimits_{r>0}{ f_{R_j}(r) \mathbb{P}[h>{\T} r^{\alpha}(\sigma^2+I_{r(j)})]\, \mathrm{d}r},
\end{eqnarray}
where $\mathbb{P}[h>Tr^{\alpha}(\sigma^2+I_{r(j)})]=e^{-\mu Tr^{\alpha}\sigma^2}\mathcal{L}_{I_{r(j)}}(\mu Tr^{\alpha})
\stackrel{(a)}{=}e^{-\mu Tr^{\alpha}\sigma^2}\mathcal{L}_{I_{r(j)}^c}(\mu Tr^{\alpha})\mathcal{L}_{I_{r(j)}^u}(\mu Tr^{\alpha})$, where ($a$) follows from independence of $I_{r(j)}^c$ and $I_{r(j)}^u$. The Laplace transform of $I_{r(j)}^c$ is given as follows:
\begin{eqnarray}
\mathcal{L}_{I_{r(j)}^c}(s)
=\mathbb{E}\Big[\prod\nolimits_{l\in \Phi_{t(j)}^c/b_0}\exp{(-sg_lR_l^{-\alpha})}\Big]
\stackrel{(a)}{=}e^{-2\pi\lambda_j\int\nolimits_{r}^{\infty}{ (1-\mathbb{E}[\exp{(-sgu^{-\alpha})}])u\, \mathrm{d}u}},\nonumber
\end{eqnarray}
where $(a)$ follows from the iid distribution of $g_l$, and its independence from $\Phi$, and the probability generating functional ($\PGFL$) of the $\PPP$ \cite{Stoyan1996}. Similarly, the Laplace transform of $I_{r(j)}^u$ is
\begin{eqnarray}
\mathcal{L}_{I_{r(j)}^u}(s)
=\mathbb{E}\Big[\prod\nolimits_{l\in \Phi_{t(j)}^u}\mathbb{E}[\exp{(-sgR_l^{-\alpha})}]\Big]
=e^{-2\pi(\lambda_t-\lambda_j)\int\nolimits_{0}^{\infty}{ (1-\mathbb{E}[\exp{(-sgu^{-\alpha})}])u\, \mathrm{d}u}}.\nonumber
\end{eqnarray}

If the fading is Rayleigh with parameter $\mu$, then the Laplace transforms of $I_{r(j)}^c$ and $I_{r(j)}^u$ equal 
\begin{eqnarray}
\label{LTs}
\mathcal{L}_{I_{r(j)}^c}(s)=e^{-2\pi\lambda_j\int\nolimits_{r}^{\infty}{ \big(\frac{1}{1+s^{-1}u^{\alpha}}\big)u\, \mathrm{d}u}},\quad
\mathcal{L}_{I_{r(j)}^u}(s)=e^{-2\pi(\lambda_t-\lambda_j)\int\nolimits_{0}^{\infty}{ \big(\frac{1}{1+s^{-1}u^{\alpha}}\big)u\, \mathrm{d}u}}.
\end{eqnarray}

Thus, the probability of coverage of a typical randomly located user is
\begin{eqnarray}
\label{generalmultifilecoverage}
{\PCOV}({\T},\lambda_j,\alpha)=\int\nolimits_{r>0}{ e^{-\lambda_j\pi r^2} e^{-\mu {\T} r^\alpha\sigma^2}\mathcal{L}_{I_{r(j)}^c}(\mu {\T} r^\alpha)\mathcal{L}_{I_{r(j)}^u}(\mu {\T} r^\alpha)2\pi\lambda_j r\, \mathrm{d}r},
\end{eqnarray}
which is obtained using (\ref{probcoverage}), (\ref{LTs}), a change of variables $v=s^{-\frac{2}{\alpha}}u^{2}$, and letting $s=\mu Tr^{\alpha}$, and the final result is obtained by a change of variables $v=r^2$ and the definitions of $\rho_1(\T,\alpha)$ and $\rho_2(\T,\alpha)$.

\subsection{Proof of Corollary \ref{ProbCov0}} 
\label{App:AppendixB0}
For Rayleigh fading with $\mu=1$, using (\ref{probcoverage}) and (\ref{generalmultifilecoverage}),  
the coverage of a typical user for $\alpha=4$ is 
\begin{multline}
\PCOV({\T},\lambda_j,4)=\int\nolimits_{r>0}{ e^{-\lambda_j\pi r^2} e^{-\T r^4\sigma^2}\mathcal{L}_{I_{r(j)}^c}({\T} r^4)\mathcal{L}_{I_{r(j)}^u}({\T} r^4)2\pi\lambda_j r\, \mathrm{d}r}\nonumber\\
=\int\nolimits_{r>0}{ e^{-\left(\lambda_j+\frac{\pi}{2}\lambda_t \sqrt{\T}-\lambda_j\sqrt{\T}\tan^{-1}\left(\frac{1}{\sqrt{\T}}\right)\right)\pi r^2}e^{-{\T} r^4\sigma^2}2\pi\lambda_j r\, \mathrm{d}r},
\end{multline}
where the final result stems from $\lambda_j=\lambda\gamma_1\sum\nolimits_{i\in f_r(j)}p_c(i)=\lambda_tp_j$, and a change of variables $u=r^2$, and the identity $\int\nolimits_{0}^{\infty}{ e^{-ax}e^{-bx^2} \, \mathrm{d}x}=\sqrt{\frac{\pi}{b}}\exp{\left(\frac{a^2}{4b}\right)}{\Q}\left(\frac{a}{\sqrt{2b}}\right)$, and letting $H({\T},\lambda_t,p_j)=\Big(\frac{p_j}{\sqrt{\T}}-p_j\tan^{-1}\big(\frac{1}{\sqrt{\T}}\big)+\frac{\pi}{2}\Big)\frac{\pi \lambda_t}{\sqrt{2\sigma^2}}$. The relation $\PCOV({\T},\lambda_j,4)$ depends on the receiver state $j$ through $\lambda_j=\lambda_tp_j$.

\end{appendix}

\begin{spacing}{1.29}
\bibliographystyle{IEEEtran}
\bibliography{D2Dreferences}
\end{spacing}

\end{document}